\documentclass[journal]{IEEEtran}
\usepackage{amsmath, amssymb,amsthm,cite}
\usepackage[dvips]{graphics}
\usepackage[utf8x]{inputenc}
\usepackage{graphicx}
\usepackage{psfrag}
\usepackage{bbm}
\usepackage{subfigure}
\usepackage{dblfloatfix}
\usepackage{mathtools}
\usepackage{algpseudocode}
\usepackage{algorithm}
\usepackage{hyperref}
\usepackage{color}
\usepackage{tcolorbox}
\usepackage{etoolbox}
\usepackage{mleftright}
\DeclareMathOperator*{\argmax}{\arg\max} 
\DeclareMathOperator*{\nn}{\nonumber}

\allowdisplaybreaks
\newcommand{\RNum}[1]{\uppercase\expandafter{\romannumeral #1\relax}}
\newtheorem{lemma}{Lemma}
\newtheorem{theorem}{Theorem}

\theoremstyle{definition}
\newtheorem{remark}{Remark}

\newtheorem{definition}{Definition}
\makeatletter
\def\blfootnote{\gdef\@thefnmark{}\@footnotetext}
\makeatother
\def\cX{{\mathcal X}}
\def\cY{{\mathcal Y}}
\def\cS{{\mathcal S}}

\def\cQ{{\mathcal Q}}

\newcommand{\pr}[1]{\left(#1\right)}
\DeclarePairedDelimiterX{\infdivx}[2]{(}{)}{%
  #1\;\delimsize\|\;#2%
}
\newcommand{\infdiv}{D\infdivx}

\title{Computable Upper Bounds on the Capacity of Finite-State Channels}
\author{Bashar~Huleihel,~\IEEEmembership{Student~Member,~IEEE,}
        Oron~Sabag,~\IEEEmembership{Member,~IEEE,}
        Haim~H.~Permuter,~\IEEEmembership{Senior~Member,~IEEE,}
        Navin~Kashyap,~\IEEEmembership{Senior~Member,~IEEE,}
        and~Shlomo~Shamai~(Shitz),~\IEEEmembership{Life~Fellow,~IEEE}}

\begin{document}
\maketitle
\begin{abstract}
We consider the use of the well-known dual capacity bounding technique for deriving upper bounds on the capacity of indecomposable finite-state channels (FSCs) with finite input and output alphabets. In this technique, capacity upper bounds are obtained by choosing suitable test distributions on the sequence of channel outputs. We propose test distributions that arise from certain graphical structures called $Q$-graphs. As we show in this paper, the advantage of this choice of test distribution is that, for the important classes of unifilar and input-driven FSCs, the resulting upper bounds can be formulated as a dynamic programming (DP) problem, which makes the bounds tractable. We illustrate this for several examples of FSCs, where we are able to solve the associated DP problems explicitly to obtain capacity upper bounds that either match or beat the best previously reported bounds. For instance, for the classical trapdoor channel, we improve the best known upper bound of $0.661$ (due to Lutz (2014)) to $0.584$, shrinking the gap to the best known lower bound of $0.572$, all bounds being in units of bits per channel use.
    
\end{abstract}

\blfootnote{This work was supported in part by the DFG via the German Israeli Project Cooperation (DIP), in part by the Israel Science Foundation (ISF), in part by the Cyber Center at Ben-Gurion University of the Negev, and in part by the WIN consortium via the Israel minister
of economy and science. The work of O. Sabag has been partially supported by the ISEF postdoctoral fellowship. The work of S. Shamai has also been supported by the European Union's Horizon 2020 Research And Innovation Programme, grant agreement no. 694630. The work of N. Kashyap was supported in part by a MATRICS grant (no. MTR/2017/000368) administered by the Science and Engineering Research Board (SERB), Govt. of India. This paper was presented in part at the 2019 IEEE
International Symposium on Information Theory \cite{8849776}.

B. Huleihel and H. H. Permuter are with the Department
of Electrical and Computer Engineering, Ben-Gurion University of the
Negev, Beer-Sheva 84105, Israel (e-mail: basharh@post.bgu.ac.il; haimp@post.bgu.ac.il).

O. Sabag is with the Department of Electrical Engineering, California Institute of Technology, Pasadena, CA 91125 USA (e-mail: oron@caltech.edu).

N. Kashyap is with the Department of Electrical Communication Engineering, Indian Institute of Science, Bangalore 560012, India (nkashyap@iisc.ac.in).

S. Shamai is with the Department of Electrical Engineering, Technion–Israel Institute of Technology, Haifa 3200003, Israel (e-mail:
sshlomo@ee.technion.ac.il).}

\begin{IEEEkeywords}
Channel capacity, dual capacity bound, dynamic programming (DP), finite state channels (FSCs).
\end{IEEEkeywords}

\section{Introduction}\label{sec:intro}
A finite-state channel (FSC) is a mathematical model for a discrete-time channel in which the channel output depends statistically on both the channel input and an underlying channel state, the latter taking values in a finite set. This model can represent a channel with memory since it allows the channel output to depend on past inputs and outputs via the channel state. In this paper, we investigate two important classes of FSCs, namely, unifilar and input-driven FSCs.

Finding a computable characterization of the capacity of these fundamental channels is a long-standing open problem in information theory. The investigation of FSCs dates back to classical works from the 1950s \cite{McMillan1953TheBT,Shannon_FSC,Blackwell58}. Besides their theoretical importance, these channels appear in many practical applications of wireless communication \cite{FSC_Wirless1,FSC_Wirless2}, and magnetic recording \cite{FSC_Magnetic}. 
Except for a few special cases where a closed-form single-letter capacity formula can be obtained, for general FSCs, only a multi-letter capacity formula exists \cite{Blackwell58,Gallager68}.

This paper advances the research on FSCs by providing a new technique to derive simple, analytical upper bounds on their capacity. For instance, consider the trapdoor channel (Fig. \ref{fig:TrapdoorChannel}) that was introduced by David Blackwell in 1961 \cite{Blackwell_trapdoor}. While its zero-error capacity \cite{Ahl_kaspi87,Ahlswede98} and its feedback capacity \cite{PermuterCuffVanRoyWeissman08} are known exactly, its channel capacity (without feedback and allowing a vanishingly small error probability) is still unknown. The best lower and upper bounds known are from \cite{kobayashi2003} and \cite{Trapdoor_Lutz}, respectively:
$$0.572\leq \mathsf{C} \leq 0.661,$$ 
where the capacity is measured in bits per channel use. In this work, we will show a novel upper bound, $\mathsf{C} \leq \log_2\left(\frac{3}{2}\right)(\approx 0.5849),$ that improves significantly upon the previous best upper bound. We will establish a general technique by which such specific bounds are relatively easy to obtain.

\begin{figure}[t]
    \centering
    \psfrag{B}[][][1]{$1$}
    \psfrag{C}[][][1]{$0$}
    \psfrag{D}[][][1]{$1$}
    \psfrag{E}[][][1]{$0$}
    \psfrag{F}[][][1]{$1$}
    \psfrag{G}[][][1]{$1$}
    \psfrag{J}[][][1]{$x_{t+2}$}
    \psfrag{K}[][][1]{$x_{t+1}$}
    \psfrag{L}[][][1]{$x_t$}
    \psfrag{M}[][][1]{$s_{t-1}$}
    \psfrag{N}[][][1]{$y_{t-1}$}
    \psfrag{O}[][][1]{$y_{t-2}$}
    \psfrag{Q}[][][1]{Channel}
    \includegraphics[scale = 1]{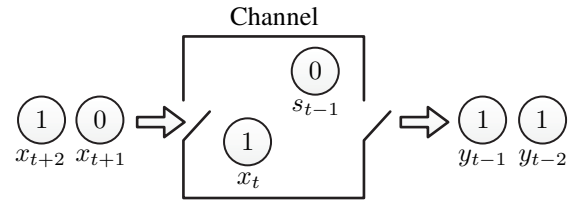}
    \caption{The trapdoor channel. The channel starts with a ball $s_{t-1}$ already in it, while a new ball $x_t$ is inserted. The channel output $y_t$ is $s_{t-1}$ or $x_t$ with equal probability, and the new channel state is the remaining ball.}
    \label{fig:TrapdoorChannel}
\end{figure}

Our upper bounds are based on a known technique called the dual capacity bounding technique, attributed to Tops{\o}e \cite{Topsoe67} --- see \cite[p.\ 147, Problem~1] {Dual_capacity}. This technique was used in \cite{MIMO_dual,Amos_Dual,Dual_Andrew_J,Duality_Kremer,Andrew_ISI,Vontobel_Dual} to obtain upper bounds for channel capacity in various contexts. In this technique, an upper bound on capacity is obtained by specifying a test distribution on the channel output process. The resulting bound is tight if the chosen test distribution is equal to the output distribution induced by the capacity-achieving input distribution. For an FSC, this output distribution is, in general, not i.i.d.. As a result, it is important to develop a systematic means of specifying a test distribution that has memory but which gives rise to a computable upper bound.

A standard choice of test distribution for channels with memory are Markov distributions of some finite order \cite{Dual_Andrew_J,Duality_Kremer,Vontobel_Dual}. However, we will use test distributions that belong to a more general class of finite-state processes. The distributions we consider are defined by a (strongly) connected\footnote{Here, by ``(strongly) connected'', we mean that between any pair of nodes $u,v$, there is a directed path from $u$ to $v$, and vice versa.} directed graph on finitely many nodes, in which each edge is labeled by a symbol from the channel output alphabet in such a way that the outgoing edges from any given node get distinct labels. For each node of the graph, we specify a probability distribution on the set of its outgoing edges. Then, walks on the graph starting from some distinguished initial node form a random process over the channel output alphabet. Following \cite{Sabag_UB_IT}, we call the underlying labeled directed graph a \emph{$Q$-graph}.\footnote{In \cite{Sabag_UB_IT,Sabag_twoway}, $Q$-graphs were used to specify mappings from channel output sequences into a finite set using directed graphs. The letter `$Q$' stands for `Quantized', as the set of nodes of the graph may be viewed as a quantization or binning of finite-length strings over the edge-label alphabet, i.e., over the output alphabet of the channel.} Note that the random process defined in this manner is a finite-state process, but it need not be Markov of any fixed order. On the other hand, it is easy to see that any Markov process of fixed order, say $m$, over a finite alphabet $\mathcal{A}$ can be defined on a certain $Q$-graph with $|\mathcal{A}|^m$ nodes, the set of nodes being in one-to-one correspondence with the set of strings of length $m$ over $\mathcal{A}$. As we will demonstrate, there is utility in going from the class of Markov test distributions to the more general class of test distributions defined on $Q$-graphs. For the specific case of the dicode erasure channel, we will show that a $Q$-graph on $3$ nodes yields an output distribution that outperforms all Markov distributions of order up to $2$.

For an FSC, the dual capacity upper bound obtained from a given test distribution is, in general, a multi-letter expression. One of the main theoretical contributions of our paper is showing that, for any test distribution defined on a $Q$-graph, the evaluation of this multi-letter expression can be formulated as an infinite-horizon average-reward dynamic programming (DP) problem. This formulation immediately gives us numerical as well as analytical tools to compute the multi-letter expression, thus yielding an explicit upper bound on capacity. Indeed, a well-known approach to handling DP optimization problems is by solving the associated Bellman equation --- see e.g., \cite{Arapos93_average_cose_survey}. Computer-based simulations of the dynamic program provide important insights into the solution of this equation.

In this paper, we use $Q$-graph based test distributions to bound from above the capacity of several well-known FSCs, namely, the trapdoor \cite{Blackwell_trapdoor}, Ising \cite{Berger90IsingChannel}, Previous Output is STate (POST) \cite{POSTchannel}, and dicode erasure \cite{henry_dissertation} channels. For each of these channels, we use the insights gained from numerical methods to arrive at an explicit analytical solution to the corresponding average-reward DP problem. In this manner, we obtain upper bounds on the capacities of these channels.

The relationship between channel capacity and DP was first observed in Tatikonda’s thesis \cite{Tatikonda00}, where it was shown that the \emph{feedback capacity} of a class of FSCs can be formulated as a DP problem. This approach was further developed in  \cite{Yang05,TatikondaMitter_IT09,PermuterCuffVanRoyWeissman08}, and yielded several new feedback capacity results for FSCs  \cite{PermuterCuffVanRoyWeissman08,Ising_channel,Sabag_BEC,Sabag_BIBO,OronBasharfeedback,PeledSabag}. However, in the case of capacity without feedback, except for the POST channel \cite{POSTchannel}, exact results are known only for certain FSCs with strict symmetry conditions, all with an i.i.d. capacity-achieving input distribution \cite{Gillbert_Mush,Gillbert_style,Song_Burst}.

The remainder of this paper is organized as follows. Section \ref{sec:prelimi} introduces our notation and defines the model. Section \ref{sec:main_results} introduces the dual capacity upper bound, gives some background on $Q$-graphs, and states our main result. Section \ref{DP_Main} gives a brief review of infinite-horizon DP and introduces the DP formulation of the dual capacity upper bound for FSCs. Section \ref{sec:analytic} presents our bounds on capacity for several specific FSCs. Finally, our conclusion appears in Section \ref{sec:conclusion}. To preserve the flow of the presentation, most of the proofs are given in the appendices.

\section{Notation and Model Definition}\label{sec:prelimi}
In this section, we introduce our notation and define our FSC model.

\subsection{Notation}
Throughout this paper, we use the following notations. The set of natural numbers (which does not include $0$) is denoted by $\mathbb{N}$, while $\mathbb{R}$ denotes the set of real numbers. Random variables will be denoted by capital letters and their realizations will be denoted by lower-case letters, e.g., $X$ and $x$, respectively. Calligraphic letters denote sets, e.g.,  $\mathcal{X}$. We use the notation $X^n$ to denote the random vector $(X_1,X_2,\dots,X_n)$ and $x^n$ to denote the realization of such a random vector. For a real number $\alpha\in[0,1]$, we define $\bar{\alpha}=1-\alpha$. The binary entropy function is denoted by $H_2(\alpha) = -\alpha\log_2(\alpha)-\bar{\alpha}\log_2(\bar{\alpha})$, where $\alpha\in[0,1]$. The probability mass function (pmf) of $X$ is denoted by $P_X$, the conditional probability of $X$ given $Y$ is denoted by $P_{X|Y}$, and the joint distribution of $X$ and $Y$ is denoted by $P_{X,Y}$. 
The probability $\Pr[X=x]$ is denoted by $P_X(x)$. When the random variable is clear from the context, we write it in shorthand as $P(x)$. 
For a conditional pmf $P_{Y|X}$, $P_{Y|X} \succ 0$ denotes that $P_{Y|X}(y|x)>0$ for all $x\in\mathcal{X}$ and $y\in\mathcal{Y}$.

Let $P_Y$ and $R_Y$ be two discrete probability measures on the same probability space. Then, $P_Y\ll R_Y$ denotes that $P_Y$ is absolutely continuous with respect to $R_Y$. The relative entropy between $P_Y$ and $R_Y$ is denoted by $D\left(P_Y\|R_Y\right)$. The conditional relative entropy is defined as $D(P_{Y|X}\|R_Y|P_X) = \mathbb{E}_X \left\{D(P_{Y|X}\|R_Y)\right\}$, where $\mathbb{E}_X\{\cdot\}$ denotes the expectation operator over $X$.

\begin{figure}[t]
\centering
    \psfrag{A}[][][0.92]{Encoder}
    \psfrag{C}[][][0.92]{Decoder}
		\psfrag{B}[][][0.85]{$P_{S^+,Y|X,S}$}
    \psfrag{E}[][][0.87]{$M$}
    \psfrag{F}[][][0.87]{$X^n$}
    \psfrag{G}[][][0.87]{$Y^n$}
		\psfrag{J}[][][0.87]{$\hat{M}$}
    \includegraphics[scale = 0.52]{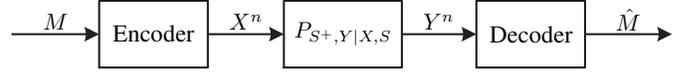}
    \caption{A finite-state channel (FSC) setting.}
    \label{fig:channelI}
\end{figure}
\subsection{FSCs}
We consider the standard finite-state channel, described in Fig. \ref{fig:channelI}. The channel is defined with finite input and output alphabets $\cX$ and $\cY$, respectively, and a finite set of states $\cS$. The input, output and state at time $t$ are denoted by $x_t$, $y_t$ and $s_t$, respectively. The defining property of an FSC is that, given $x_t$ and $s_{t-1}$, the pair $(s_t,y_t)$ is conditionally independent of all previous inputs, outputs and states, as well as of the message $m$ to be transmitted. To be precise,
\begin{align}\label{eq:FSC}
    P(s_t,y_t|x^t,s^{t-1},y^{t-1},m) = P_{S^+,Y|X,S}(s_t,y_t|x_t,s_{t-1}),
\end{align}
where $S$ denotes the channel state at the beginning of the transmission and $S^+$ represents the channel state at the end of the transmission. In particular, the transition probability kernel $P_{S^+,Y|X,S}$ is time-invariant, i.e., it does not depend on $t$. 
Furthermore, if there is no feedback, the conditional probability $P_{S^t,Y^t|X^t,S_0}$ decomposes as
\begin{align*}
    P_{S^t,Y^t|X^t,S_0}(s^t,y^t|x^t,s_0) = \prod_{i=1}^t P_{S^+,Y|X,S}(s_i,y_i|x_i,s_{i-1}).
\end{align*}

The following definition presents the indecomposability property of FSCs.
\begin{definition} \label{def:ind}
[\!\cite{Gallager68}, Ch. 4.6] An FSC is \emph{indecomposable} if for any $\epsilon>0$, there exists an $N$ such that, for all $n\geq N$,
\begin{align} \label{eq: indec: property}
	|P(s_n|x^n,s_0)-P(s_n|x^n,s^{\prime}_0)|\leq \epsilon
\end{align}
for any channel states $s_n$, $s_0$, $s^{\prime}_0$, and any input sequence $x^n$.
\end{definition}

Loosely speaking, for an indecomposable FSC, the effect of the initial channel state becomes negligible as time evolves. An alternative characterization of indecomposability \cite[Theorem~4.6.3]{Gallager68} is that for some $n$ and each input sequence $x^n \in \cX^n$, there is a choice of state $s_n$ at time $n$ ($s^n$ may depend on $x^n$) such that $P(s_n | x^n,s_0) > 0$ for all initial states $s_0$.

The capacity of an indecomposable channel is presented in the following theorem.
\begin{theorem}[\!\cite{Gallager68}, Ch. 4.6] \label{FSC_Capacity}
 The capacity of an indecomposable FSC is
\begin{align*}
	\mathsf{C} = \lim_{n\to\infty}\max_{P(x^n)}\frac{1}{n}I\pr{X^n;Y^n|S_0 = s_0},
\end{align*}
for any $s_0\in \mathcal{S}$.
\end{theorem}

Throughout this paper, the capacity (and bounds on it) are measured in bits per channel use. We investigate the following two important classes of FSCs:
\begin{enumerate}
    \item \textbf{Unifilar FSCs:} For these channels, the state evolution is given by a deterministic function. Specifically, \eqref{eq:FSC} is simplified to:
    \begin{align}
        &P_{S^+,Y|X,S}(s_t,y_t|x_t,s_{t-1}) \nn\\&= \mathbbm{1}\{s_t=f(x_t,y_t,s_{t-1})\}P_{Y|X,S}(y_t|x_t,s_{t-1}),
    \end{align}
    where $f:\mathcal X\times \mathcal{Y}\times\mathcal{S}\to \mathcal S$. Since the channel state can be computed recursively, we may use $s_t = f^t(x^t,y^t,s_0)$ to denote $t$ applications of $f(\cdot)$. 
    \item \textbf{Input-driven FSCs:} For these channels, the channel state does not depend on past outputs. Specifically,
       \begin{align}\label{eq:input_def}
        &P_{S^+,Y|X,S}(s_t,y_t|x_t,s_{t-1}) \nn\\&= P_{S^+|X,S}(s_t|x_t,s_{t-1})P_{Y|X,S}(y_t|x_t,s_{t-1}).
    \end{align}
    Note that this definition generalizes that of FSCs with input-dependent states \cite{FSC_Neri}, in which the next state is a deterministic function of the input and the previous state.
\end{enumerate}

\section{Main Result via Dual Capacity Formula} \label{sec:main_results}
In this section, we present the dual capacity upper bound, $Q$-graphs and our main result.
\subsection{Dual capacity upper bound}
The dual capacity upper bound \cite{Topsoe67,Dual_capacity} is a simple upper bound on channel capacity that has been utilized in many works \cite{MIMO_dual,Amos_Dual,Dual_Andrew_J,Duality_Kremer,Andrew_ISI,Vontobel_Dual}. For any memoryless channel, $P_{Y|X}$, and test distribution $R_Y$, the dual capacity upper bound is given by
\begin{align}\label{DB_MC}
\mathsf{C}\leq \max_{x\in\mathcal{X}} D\left(P_{Y|X=x} \middle\| R_{Y}\right).
\end{align}
The proof follows from the following steps:
\begin{align}\label{DB_step}
I(X;Y)&= D\left(P_{Y|X}\middle\|R_Y \middle|P_X \right) - D\left(P_Y\|R_Y\right)\nn
\\&\leq D\left(P_{Y|X}\middle\|R_Y \middle|P_X \right)\nn
\\&\leq \max_{x\in\mathcal{X}}D\left(P_{Y|X=x}\middle\|R_Y \right).
\end{align}
The bound is tight if $R_Y$ is equal to the output distribution, $P_Y^*$, induced by an optimal (i.e., capacity-achieving) input distribution.

For FSCs, where the aim is to maximize the $n$-letter mutual information $I(X^n;Y^n)$, one may replace the test distribution in \eqref{DB_step} with $R_{Y^n}$ and obtain
\begin{align} \label{co.DB}
I(X^n;Y^n)&\leq \max_{x^n\in\mathcal{X}^n} D\left(P_{Y^n|X^n=x^n}\middle\|R_{Y^n}\right). 
\end{align}
Again, this bound is tight when $R_{Y^n}=P_{Y^n}^*$, the output distribution induced by an input distribution that maximizes $I(X^n;Y^n)$. Naturally, the choice of the test distribution will affect the tightness of the bound, and we would like to choose test distributions that are close, in some sense, to $P_{Y^n}^*$. The output distribution $P_{Y^n}^*$ is, in general, not i.i.d.. A common choice of a test distribution is a Markov distribution of some finite order \cite{Dual_Andrew_J,Duality_Kremer,Vontobel_Dual}, but here we use an extension of this notion. The mathematical structure needed to define this extension is called a $Q$-graph, which is presented in the next section.

\subsection{The $Q$-graph}
A $Q$-graph, introduced in \cite{Sabag_UB_IT}, is a directed and (strongly) connected graph on a finite set of nodes $\cQ$, in which each node has $|\mathcal{Y}|$ outgoing edges with distinct labels. Due to the distinct labeling, the graph defines a mapping $\phi:\cQ\times\cY\to \cQ$, where $\phi(q,y)$ is the unique node pointed to by the edge from $q$ labeled with $y$. Further, given a distinguished initial node $q_0 \in \cQ$, we also have a well-defined mapping $\Phi_{q_0}:{\cY}^{*}\to \cQ$, where $\cY^*$ is the set of all finite-length sequences over $\cY$. Indeed, $\Phi_{q_0}(y^t)$ is the node reached by walking along the unique directed path of length $t$ starting from $q_0$ and labeled by $y^t = (y_1,y_2,\ldots,y_t)$. We will often drop the subscript from $\Phi_{q_0}$ for notational convenience, whenever there is no ambiguity in doing so.

\begin{figure}[tb]
\centering
    \psfrag{Q}[][][0.9]{$\;Y=1$}
    \psfrag{E}[][][0.9]{$\;\;Y=0$}
    \psfrag{L}[][][0.9]{$Q=1$}
    \psfrag{H}[][][0.9]{$Q=0$}
    \includegraphics[scale = 0.5]{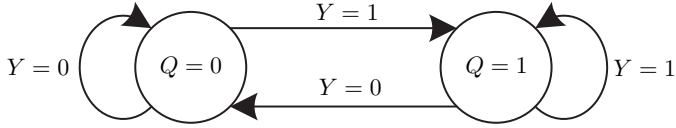}
    \caption{A $1$st-order Markov $Q$-graph for channel output alphabet $\mathcal{Y}=\{0,1\}$.}
    \label{fig:1Markov}
\end{figure}

Fix a $Q$-graph on the set of nodes $\cQ$, with a distinguished initial node $q_0$. A graph-based test distribution, $R_{Y|Q}$, is a collection of probability distributions $R_{Y|Q=q}$ on $\cY$, defined for each $q \in \cQ$. This defines a test distribution on channel output sequences as follows:
\begin{align}\label{eq:test_dist}
	R_{Y^n|q_0}(y^n) = \prod_{t=1}^n R_{Y|Q}(y_t|q_{t-1}),
\end{align}
where $q_{t-1}=\Phi(y^{t-1})$ for $t > 1$. It can be noted from \eqref{eq:test_dist} that, since $|\mathcal Q|<\infty$, the induced process is a finite-state process.

A special case of a $Q$-graph is a \emph{$k$th-order Markov $Q$-graph}, which is defined on the set of nodes $\cQ = \cY^k$, and for each node $q = (y_1,y_2,\ldots,y_k)$, the outgoing edge labeled $y \in \cY$ goes to the node $(y_2,\ldots,y_k,y)$. For instance, Fig.~\ref{fig:1Markov} shows a Markov $Q$-graph with $\mathcal{Y} = \{ 0, 1 \}$ and $k=1$. Note that test distributions $R_{Y | Q}$ on a $k$th-order Markov $Q$-graph correspond to $k$th-order stationary Markov processes.

$Q$-graph-based test distributions grant us an added layer of generality over Markov distributions of finite order. There is value to this added generality, as we will see in Section~\ref{sec:DEC}.
Moreover, the dual capacity upper bound obtained from any such test distribution is actually computable (at least numerically) for certain classes of FSCs.   
Indeed, our main result is that, for unifilar and input-driven FSCs, the dual capacity upper bound obtained from any $Q$-graph based test distribution can be formulated as a DP problem, and hence, is computable.

\subsection{Summary of main results}
In this section we summarize our main contributions for unifilar and input-driven FSCs.\\
Specifically, our main contributions are as follows:   
\begin{itemize}
\item In Section \ref{DP_Main}, we derive the duality upper bounds for unifilar FSCs and input-driven FSCs in Theorem \ref{main_bound} and Theorem \ref{main_bound2}, respectively. The duality bounds hold for any choice of a graph-based test distribution and are given by multi-letter expressions, i.e., they depend on a limiting blocklength.
\item In Section \ref{DP_Main}, we show the computability of the bounds by formulating them as a DP. Specifically, when the FSC is unifilar, we show that the dual capacity upper bound in Theorem \ref{main_bound} can be formulated as a dynamic program with $\cal{P}(\mathcal{S}\times\mathcal{Q})$ being the state space and $\mathcal{X}$ being the action space.
\item Similarly, if the channel is an input-driven FSC, then the dual capacity upper bound in Theorem \ref{main_bound2} can be formulated as a dynamic program with $\cal{P}(\mathcal{S})\times \cal{P}(\mathcal{Q})$ being the state space and $\mathcal{X}$ being the action space.
\item In Section~\ref{sec:analytic}, we apply the developed framework to several examples and derive novel upper bounds on the capacity of the well-known trapdoor and Ising channels that outperform previously reported upper bounds. Further, we provide an alternative converse proof for the capacity of the POST channel.
\item Lastly, in Section~\ref{sec:analytic}, we demonstrate the superiority of the graph-based test distribution compared to simple Markovian test distributions by comparing the duality upper bound for the DEC. 
\end{itemize}

In the next section, we introduce the DP framework and formally define the DP formulations stated above. The DP formulations are useful as we can then use known DP algorithms to numerically compute upper bounds on capacity. Moreover, the numerical results can sometimes be converted to explicit analytical upper bounds, as we do for the examples presented in Section~\ref{sec:analytic}.

The dual capacity bounding technique has been utilized in several works, e.g., for amplitude-constrained additive white Gaussian noise channels \cite{Duality_Kremer,Rassouli_dual}. In \cite{Loeliger_FSCs, FSC_num_bounds}, the authors derive bounds on the capacity of channels with memory and provide numerical methods to approximate the bounds. Our work is closest in spirit to that in \cite{Dual_Andrew_J} and \cite{Andrew_ISI}, in which the dual capacity bounding technique is applied to binary-input memoryless channels with a runlength constrained input and to single-tap binary-input Gaussian channels with intersymbol interference. Using Markov test distributions, the authors of \cite{Dual_Andrew_J} and \cite{Andrew_ISI} are able to derive, in some specific cases, explicit expressions for the resulting upper bounds on channel capacity.

The main novelty in our work is the DP formulation of the dual capacity upper bound and the use of graph-based test distributions. On the one hand, our formulation is restricted to channels with finite input, output, and state alphabets, but on the other hand, it allows us to use the powerful machinery of DP to at least numerically evaluate the bounds for a large class of FSCs. In some cases, as we will see in Section~\ref{sec:analytic}, we are even able to convert the numerically evaluated bounds to analytical expressions.

\section{Upper bounds via DP} \label{DP_Main} 
In this section, we first introduce DP and the Bellman equation. 
Then, for a fixed graph-based test distribution, we present a DP formulation of the dual capacity upper bound for unifilar and input-driven FSCs. Additionally, we present a simplified DP formulation for the case of unifilar input-driven FSCs, where the state evolves according to $s_t=f(x_t,s_{t-1})$.

\subsection{DP and the Bellman equation}
Here we introduce a formulation for a deterministic\footnote{The DP formulation we consider is \emph{deterministic}, in the sense that we do not introduce a (random) disturbance in the formulation.} average-reward dynamic program. Each DP problem is defined by a quintuple $\left(\mathcal{Z},\mathcal{U},F,P_Z,g\right)$. 
Each action, $u_t$, takes a value in a compact subset $\mathcal{U}$ of a Borel space. We consider a discrete-time dynamical system that evolves according to 
\begin{align*}
z_t = F(z_{t-1},u_t), \;\;\; t = 1,2,3,\dots,
\end{align*}
where each DP state, $z_t$, takes values in a Borel space $\mathcal{Z}$. The initial state $z_0$ is drawn according to the distribution $P_Z$.  
The action $u_t$ is selected by a deterministic function $\mu_t$ that maps the initial DP state, $z_0$, into actions. Specifically, given a policy $\pi=\{\mu_1,\mu_2,...\}$, actions are generated according to $u_t = \mu_t(z_0)$. Accordingly, in this setup, the only randomness is in $z_0$.

Given a bounded reward function, $g:\mathcal{Z}\times\mathcal{U}\rightarrow\mathbb{R}$, we aim to maximize the average reward. The average reward for a policy $\pi$ is defined by $\rho_\pi = \liminf_{n\to\infty}\frac{1}{n}\mathbb{E}_\pi\left[\sum_{t=0}^{n-1}g\left(Z_t,\mu_{t+1}(z_0)\right)\right]$, where the subscript $\pi$ indicates that actions are generated by the policy $\pi=(\mu_1,\mu_2,...)$. The optimal average reward is given by $\rho^*=\sup_{\pi}\rho_\pi$.   

The following theorem, an immediate consequence of Theorem $6.1$ in \cite{Arapos93_average_cose_survey}, encapsulates the Bellman equation, which provides a sufficient condition for the optimality of an average reward and a policy.

\begin{theorem}[Bellman equation]
\label{Theorem:Bellman}
Given a DP problem as above, if a scalar $\rho\in\mathbb{R}$ and a bounded function $h:\mathcal{Z}\rightarrow\mathbb{R}$ satisfy
\begin{align*}
    \rho+h(z) = \sup_{u\in\mathcal{U}}\pr{g\pr{z,u}+h\pr{F\pr{z,u}}},\;\; \forall z\in\mathcal{Z}
\end{align*}
then $\rho=\rho^{*}$.
\end{theorem}

Numerical methods for solving the DP problem, such as policy iteration and value iteration, provide very important insights into the solution of the Bellman equation. One may use the approximate solution obtained by these algorithms to generate a conjecture for the exact solution, and use the Bellman equation to verify its optimality.


\subsection{A DP formulation for unifilar FSCs} \label{DP_main1} 
In this section we introduce a DP formulation of the dual upper bound on the capacity of unifilar FSCs. First, let us present a definition that extends the idea of channel indecomposability (see Definition \ref{def:ind} and its subsequent paragraph) to the notion of a channel and a graph-based test distribution being jointly indecomposable.
\begin{definition}[Joint indecomposability] \label{def:sq_ind}
Fix an FSC and a graph-based test distribution on channel output sequences. If for some $n\in\mathbb{N}$ and each input sequence $x^n$, there exists a choice of $s_n$ and $q_n$ such that
\begin{align} \label{eq: Q_cond}
    P(s_n,q_n|x^n,s_0,q_0) > 0, \text{\quad for all $s_0,q_0$,}
\end{align}
then the FSC and test distribution are \emph{jointly indecomposable}. Note that $s_n$ and $q_n$ above are allowed to depend on $x^n$.
\end{definition}

The following theorem presents an upper bound on the capacity of unifilar FSCs, which is a simplification of the dual capacity upper bound for FSCs when choosing graph-based test distributions on channel outputs.
\begin{theorem} \label{main_bound}
For any unifilar FSC and a graph-based test distribution $R_{Y|Q} \succ 0$ that are jointly indecomposable, the channel capacity is bounded as
\begin{align} \label{eq: UB_unifilar}
	\mathsf{C} \leq \lim_{n\to\infty} \max_{x^n\in\mathcal{X}^n}&\frac{1}{n}\sum_{t=1}^n\sum_{q,s}z_{t-1}(q,s)\nn\\&\times D\left(P_{Y|X,S}(\cdot|x_t,s) \middle\| R_{Y|Q}(\cdot|q)\right),
\end{align}
for any $(s_0,q_0)$, where $$z_{t-1}(q,s) \triangleq P_{Q_{t-1},S_{t-1}|X^{t-1},S_0,Q_0}(q,s|x^{t-1},s_0,q_0).$$
\end{theorem}
\begin{remark}\label{remark: UB_restriction}
In the statement of Theorem \ref{main_bound}, the condition  that $R_{Y|Q}$ be strictly positive is imposed to simplify the presentation of the result. Indeed, without this restriction, the bound may be infinite, which is still a valid upper bound. However, more importantly, $R_{Y|Q} \succ 0$ ensures that the condition $P_{Y_t|Y^{t-1}=y^{t-1},X^t=x^t}  \ll R_{Y_t|Q_{t-1}=q_{t-1}}$ holds for any $q_0$, $x^t$ and $y^{t-1}$, where $q_{t-1}=\Phi_{q_0}(y^{t-1})$. The latter condition ensures that the upper bound does not depend on the choice of the initial state. This point will be addressed precisely in the proof of Theorem \ref{main_bound} --- see Appendix~\ref{app:subsec_initialstate}.
\end{remark}

The proof of Theorem \ref{main_bound} is given in Appendix \ref{app: UB_uniflar}. 
We now present the DP formulation of the upper bound in Theorem \ref{main_bound}.

\begin{table}[t]
\caption{DP Formulation for Unifilar FSCs}
\label{table: main1}
\begin{center}
\scalebox{1.2}{
 \begin{tabular}{|c | c |} 
 \hline
 DP notations & Upper bound on capacity \\ [0.5ex] 
 \hline\hline
The DP state & $z_t = P_{Q_t,S_t|X^t=x^t,s_0,q_0}$ \\ [0.3ex]
 \hline
 The action & $u_t = x_t$ \\[0.3ex]
 \hline
The reward & Eq. \eqref{reward_1}\\[0.3ex]
 \hline
 DP state evolution & Eq. \eqref{Next_state} \\[0.3ex]
 \hline
\end{tabular} }
\end{center}
\end{table}
Throughout the derivations, a $Q$-graph and a test distribution, $R_{Y|Q}$, are fixed. At each time $t$, let the action be the current channel input $u_t \triangleq x_t$, which takes values in $\mathcal{X}$. The DP state, $z_{t-1}$, is defined as appears in Theorem \ref{main_bound}.
The reward function is defined by
\begin{align} \label{reward_1}
	g(z_{t-1},u_t) \triangleq &\sum_{q,s}z_{t-1}(q,s) D\left(P_{Y|X,S}(\cdot|u_t,s)\middle\|R_{Y|Q}(\cdot|q)\right).
\end{align}
The above formulation is summarized in Table \ref{table: main1}. We show in Appendix~\ref{app: formulation_unifilar}, as part of the proof of Theorem \ref{theorem: formulation}, that this constitutes a valid DP. The infinite-horizon average reward of this DP is given by
\begin{align} \label{eq: optimal_reward_1}
    &\rho^*=\sup_{\{x_i\}_{i=1}^{\infty}}\liminf_{n\to\infty}\frac{1}{n}\sum_{t=0}^{n}\sum_{q,s}z_{t-1}(q,s)\nn\\&\qquad\qquad\qquad\qquad\times D\left(P_{Y|X,S}(\cdot|x_t,s)\middle\|R_{Y|Q}(\cdot|q)\right). 
\end{align}
The following theorem summarizes the relation between the upper bound in Theorem \ref{main_bound} and $\rho^\ast$.
\begin{theorem}[DP formulation of the upper bound]\label{theorem: formulation}
The upper bound in Theorem \ref{main_bound} is equal to the optimal average reward in \eqref{eq: optimal_reward_1}. That is, the capacity is upper bounded by $\rho^*$, the optimal average reward of the DP defined above.
\end{theorem}
The proof of Theorem \ref{theorem: formulation} is given in Appendix~\ref{app: formulation_unifilar}.


\subsubsection*{\textbf{Special Case}}\label{S_DP}
We consider here a special case of the upper bound in Theorem \ref{main_bound} for which the DP formulation simplifies significantly. That is, it will be shown that both the DP state and action are discrete. 
Assume that the channel state is evaluated according to $s_t = f(x_{t},s_{t-1})$. This time we use a $k$th-order Markov $Q$-graph. For this case, for any $s_0\in\mathcal{S}$, the upper bound in Theorem \ref{main_bound} is simplified to
\begin{align}
	\mathsf{C} \leq \lim_{n\to\infty}\max_{x^n}&\frac{1}{n}\sum_{t=1}^n \sum_{y_{t-k}^{t-1}} \left(\prod_{i=t-k}^{t-1} P_{Y|X,S}(y_i|x_i,s_{i-1})\right)\nn\\&\times
	 D\left(P_{Y|X,S}(\cdot|x_t,s_{t-1})\middle\|R_{Y|Q}\left(\cdot|y_{t-k}^{t-1}\right)\right),\nn
\end{align}
where $R_{Y|Q}\left(\cdot|y_{t-k}^{t-1}\right)$ denotes a $k$th-order Markov distribution. Note that the simplification follows directly by considering $q=y_{t-k}^{t-1}$ (corresponding to a Markov $Q$-graph of $k$th-order) in Theorem \ref{main_bound}, and observing that $z_{t-1}(q,s)$ can be written as the product term within the parentheses. 

For this special case, the DP formulation is the same as that for the unifilar FSC, but the DP state simplifies to $z_{t-1} \triangleq (x_{t-k}^{t-1},s_{t-k-1})\in\mathcal{X}^k\times\mathcal{S}$ since the reward at time $t$ simplifies to
\begin{align}\label{reward_func_ising}
	g(z_{t-1},u_t) &\triangleq \sum_{y_{t-k}^{t-1}} \left(\prod_{i=t-k}^{t-1} P_{Y|X,S}(y_i|x_i,s_{i-1})\right)\nn\\&\times
	 D\left(P_{Y|X,S}(\cdot|u_t,s_{t-1})\middle\|R_{Y|Q}\left(\cdot|y_{t-k}^{t-1}\right)\right).
\end{align}
Specifically, from \eqref{reward_func_ising} and the assumption that $s_t=f(x_t,s_{t-1})$, it follows that the reward is a function of $(x_{t-k}^t,s_{t-k-1})$. Thus, it is a function of the previous DP state $z_{t-1}$ and the action $x_t$.

Note that in this formulation the DP state and the action take values from a finite set. Consequently, the numerical evaluation and the subsequent analytical derivation of the solution to the Bellman equation become more tractable.


\subsection{A DP formulation for input-driven FSCs} \label{DP_main2} 
The following theorem presents an upper bound on the capacity of an input-driven FSC, which is a simplification of the dual upper bound for FSCs when choosing graph-based test distributions.
\begin{theorem} \label{main_bound2}
For an input-driven FSC and a graph-based test distribution $R_{Y|Q} \succ 0$ that are jointly indecomposable, the channel capacity is bounded as
\begin{align*}
	\mathsf{C} \leq& \lim_{n\to\infty} \max_{x^n\in\mathcal{X}^n}\frac{1}{n}\sum_{t=1}^n\sum_{q\in \cQ}\beta_{t-1}(q) \nn\\&\times\infdiv[\Bigg]{\sum_{s_{t-1}}\gamma_{t-1}(s_{t-1})\cdot P_{Y|X,S}(\cdot|x_t,s_{t-1})}{R_{Y|Q}(\cdot|q)},
\end{align*}
for any $(s_0,q_0)$, where
\begin{align} 
	\beta_{t-1}(q) &=  P_{Q_{t-1}|X^{t-1},S_0,Q_0}(q|x^{t-1},s_0,q_0),  \label{eq: th10:state1}\\
	\gamma_{t-1}(s_{t-1}) &=  P_{S_{t-1}|X^{t-1},S_0}(s_{t-1}|x^{t-1},s_0). \label{eq: th10:state2}
\end{align}
\end{theorem}
The proof of Theorem \ref{main_bound2} is given in Appendix \ref{proof_main2}. Remark \ref{remark: UB_restriction} applies also to Theorem \ref{main_bound2}. We now present the DP formulation of the upper bound in Theorem \ref{main_bound2}, and it will be shown that this formulation satisfies the DP properties. 

\begin{table}[t]
\caption{DP Formulation for Input-driven FSCs}
\label{table: main2}
\begin{center}
\scalebox{1.08}{
 \begin{tabular}{|c | c |} 
 \hline
 DP notations & Upper bound on capacity \\ [0.5ex] 
 \hline\hline 
The DP state & $z_t = \left[P_{Q_t|X^t=x^t,s_0,q_0}, P_{S_t|X^t=x^t,s_0}\right]$ \\ [0.3ex]
 \hline
 The action & $u_t = x_t$ \\[0.5ex]
 \hline
 The reward function & Eq. \eqref{reward_2}\\[0.3ex]
 \hline
  DP state evolution & Eq. \eqref{DP4_state1}-\eqref{DP4_state2} \\[0.3ex]
 \hline
\end{tabular} }
\end{center}
\end{table}
Throughout the derivations, a $Q$-graph and a test distribution, $R_{Y|Q}$, are fixed. At each time $t$, let the action be the current channel input $u_t \triangleq x_t$. 
The DP state is defined as $z_{t-1} \triangleq (\beta_{t-1},\gamma_{t-1})$, where $\beta_{t-1}$ and $\gamma_{t-1}$ are defined in \eqref{eq: th10:state1}, \eqref{eq: th10:state2}. 
The reward function is defined by
\begin{align} \label{reward_2}
	&g(z_{t-1},u_t) \triangleq \sum_{q\in\mathcal{Q}}\beta_{t-1}(q) \nn\\&\times\infdiv[\Bigg]{\sum_{s_{t-1}}\gamma_{t-1}(s_{t-1})\cdot P_{Y|X,S}(\cdot|u_t,s_{t-1})}{R_{Y|Q}(\cdot|q)}.
\end{align}
The above formulation is summarized in Table \ref{table: main2}. Assuming that it is a valid DP, this DP formulation implies that the infinite-horizon average reward is:
\begin{align}\label{eq: optimal_reward_2}
    &\rho^*=\sup_{\{x_i\}_{i=1}^{\infty}}\liminf_{n\to\infty}\frac{1}{n}\sum_{t=0}^{n}\sum_{q\in\mathcal{Q}}\beta_{t-1}(q) \nn\\&\times\infdiv[\Bigg]{\sum_{s_{t-1}}\gamma_{t-1}(s_{t-1})\cdot P_{Y|X,S}(\cdot|x_t,s_{t-1})}{R_{Y|Q}(\cdot|q)}.
\end{align}
The following theorem provides the relation between the upper bound in Theorem \ref{main_bound2} and $\rho^\ast$.
\begin{theorem}[DP formulation of the upper bound]\label{theorem: formulation_input_dependent}
The upper bound in Theorem \ref{main_bound2} is equal to the optimal average reward in \eqref{eq: optimal_reward_2}. That is, the capacity is upper bounded by $\rho^*$, the average reward of the defined DP.
\end{theorem}
The proof of Theorem \ref{theorem: formulation_input_dependent} is given in Appendix \ref{App: formulation_input_dependent}.

\section{Examples}\label{sec:analytic}
In this section, we study four examples of FSCs. For all examples, the input and the state take values from the binary alphabet, i.e., $\mathcal{S}=\mathcal{X} = \{0,1\}$.

In Section \ref{DP_Main} we presented several DP formulations in which the action space is discrete, while the DP state space might be either discrete or continuous, depending on the channel state evolution and the choice of the test distribution. In the case where both the DP state space and the action space are discrete, numerical methods (such as policy iteration and value iteration) always yield sufficient insights to solve the Bellman equation and extract the function $h$ and the optimal reward $\rho$ analytically. Accordingly, in this case, analytic upper bounds can be easily derived. In general, however, there is no systematic way of analytically determining a $h$ and a $\rho$ that satisfy the Bellman equation. Nevertheless, in this section we present single-letter upper bounds on the capacity of several channels that were derived by solving a DP problem with a continuous DP state space while using the insights gained from the numerical methods.

In addition, one of the main challenges here is to choose $Q$-graphs that will result in tight bounds. To this end, following \cite{OronBasharfeedback2}, we create a pool of all valid $Q$-graphs up to a fixed size of nodes and choose the $Q$-graphs that result in the best upper bounds. A particular choice of a $Q$-graph is the $k$th-order Markov $Q$-graph which, as will be shown, in some cases, provides very good upper bounds.

\subsection{The Trapdoor Channel}
The trapdoor channel was introduced by David Blackwell in 1961 \cite{Blackwell_trapdoor}. Its operation proceeds as follows: at each time $t$, the channel input, $x_t$, is transmitted through the channel and the channel state is $s_{t-1}$. The channel output, $y_t$, is equal to $s_{t-1}$ or to $x_t$ with the same probability. The new channel state is evaluated according to $s_t=x_t\oplus y_t\oplus s_{t-1}$, where $\oplus$ denotes the XOR operation. Accordingly, the trapdoor channel is a unifilar FSC. An illustration of the trapdoor channel appears in Fig. \ref{fig:TrapdoorChannel}.

The zero-error capacity of the trapdoor channel was found by Ahlswede \textit{et al}. \cite{Ahl_kaspi87,Ahlswede98} and is equal to 0.5 bits per channel use. Furthermore, the feedback capacity of this channel was found in \cite{PermuterCuffVanRoyWeissman08} to be $\mathsf{C_{FB}}=\log_2{\phi}$, where $\phi$ is the golden ratio, $\frac{1+\sqrt{5}}{2}$. However, the trapdoor channel was originally introduced as a channel without feedback, and the capacity of this channel in the absence of feedback is still open. 
The best known lower and upper bounds obtained so far in the literature imply that
\begin{align} \label{eq: trap_previousBest}
    0.572\leq \mathsf{C_{trapdoor}}\leq 0.6610, 
\end{align}
where the lower bound is derived in \cite{kobayashi2003}, and the upper bound is derived in \cite{Trapdoor_Lutz}.
In the following theorem, we introduce a novel upper bound on the capacity of the trapdoor channel that significantly improves the upper bound in \eqref{eq: trap_previousBest}. 
\begin{theorem} \label{theorem: Trapdoor}
The capacity of the trapdoor channel is upper-bounded by
\begin{align*} 
\mathsf{C_{trapdoor}} \leq \log_2\left(\frac{3}{2}\right). 
\end{align*}
\end{theorem}
The value of $\log_2\left(\frac{3}{2}\right)$ is approximately $0.5849$, which concludes our new upper bound, $$0.572\le \mathsf{C_{trapdoor}}\le 0.5849.$$

The proof of Theorem \ref{theorem: Trapdoor} is presented in Appendix \ref{app. trapdoor}. It involves analytically solving the Bellman equation (Theorem \ref{Theorem:Bellman}) corresponding to the DP formulation of the bound in Theorem~\ref{main_bound} obtained from a graph-based test distribution defined on the $Q$-graph in Fig.~\ref{fig:1Markov}. The chosen test distribution, the function $h$, and the optimal average reward $\rho^*$ that are used to solve the Bellman equation are given in the appendix.

\subsection{The Ising Channel}
The Ising channel was introduced as an information-theoretic channel by Berger and Bonomi in $1990$ \cite{Berger90IsingChannel}. Resembling the well-known physical Ising model, it models a channel with intersymbol interference. The channel operates as follows. At each time $t$, the channel input, $x_t$, is transmitted through the channel while the channel state is $s_{t-1}$. The channel output, $y_t$, is equal to $s_{t-1}$ or to $x_t$ with probability $0.5$. The new channel state is $s_t=x_t$, and therefore, the channel is both unifilar and input-driven.

The feedback capacity of the Ising channel was shown in \cite{Ising_channel} to be approximately $0.5755$. In the absence of feedback, the capacity is still unknown, and the best known lower and upper bounds were recently derived in \cite{Ising_artyom_IT} and are given by
\begin{align} \label{eq: Ising_UB_old}
   0.5451\leq \mathsf{C_{Ising}}\leq 0.551.
\end{align}
In the following theorem, we introduce a novel upper bound on the capacity of the Ising channel that improves upon the upper bound in \eqref{eq: Ising_UB_old}.

\begin{theorem} \label{theorem: Ising}
The capacity of the Ising channel is upper-bounded by
\begin{align*} 
\mathsf{C_{Ising}}\leq \min \frac{1}{4}\log_2\left(\frac{1}{2acd(1-a)}\right), 
\end{align*}
where the minimum is over all $(a,b,c,d)\in(0,1)^4$ that satisfy:
\begin{align} \label{eq:Ising_constraints}
	0 &\leq 2d\bar{c}-a^2 \nn\\
	0 &\leq a^3-2\bar{a}cd \nn\\
	0 &\leq 4b\bar{c}^2\bar{d}-a^2\bar{a}c \nn\\
	0 &\leq 32b^2\bar{b}\bar{c}^2\bar{d}^2 - ac^2d^2\bar{a}^2.
\end{align}
\end{theorem}
Evaluation of the bound shows that it is equal to approximately $0.5482$. Thus, the lower bound in \cite{Ising_artyom_IT} is almost tight: $$0.5451\le \mathsf{C_{Ising}}\le 0.5482.$$

The proof of Theorem \ref{theorem: Ising} is presented in Appendix \ref{app. Ising}, and follows from analytically solving the Bellman equation (Theorem \ref{Theorem:Bellman}) while using a graph-based test distribution that is structured on a Markov $Q$-graph with $k=3$. The upper bound in Theorem \ref{main_bound} can also be evaluated for Markov graphs of higher order. However, given the elegant expression obtained by using Markov graphs of order $k=3$ and the minor improvement seen when $k$ is increased, we present only the case of $k=3$.


\subsection{The Dicode Erasure Channel} \label{sec:DEC}
The main objective of this example is to demonstrate that the notion of $Q$-graphs can indeed be useful in a search for good bounds. We will show that for a simple channel known as the dicode erasure channel (DEC), a small $Q$-graph outperforms all Markov test distributions up to order $2$. 

The DEC has been investigated in \cite{PfitserLDPC_memory_erasure, henry_dissertation, Sabag_UB_IT} and stands as a simplified version of the well-known dicode channel with white additive Gaussian noise (AWGN) used as a model in magnetic recording \cite{DEC_Magnetic}. Specifically, in response to the input sequence $(x_t)$, the DEC with parameter $\epsilon \in [0,1]$ produces as output the sequence $(y_t)$, where 
$$
y_t = \begin{cases} x_t - x_{t-1} & \text{ with probability } 1-\epsilon \\
? & \text{ with probability } \epsilon.
\end{cases}
$$
The channel state is the previous input, i.e., $s_{t-1}=x_{t-1}$, so that the channel is both unifilar and input-driven.

The feedback capacity of the DEC channel was derived in \cite{Sabag_UB_IT}. However, in the absence of feedback, the problem of determining the capacity is still open.
In the following theorem, we present an upper bound on the DEC capacity. 
\begin{theorem} \label{theorem: DEC_UB}
The capacity of the DEC with erasure probability $\epsilon \in [0,1]$ is upper-bounded by
\begin{align} 
    \mathsf{C_{DEC}}(\epsilon) \leq 1+\epsilon\log_2\left(\frac{1-p}{p}\right),
\end{align}
where $p\in(0,1)$ solves $p^\epsilon = 2\bar{p}$.
\end{theorem}

\begin{figure}[tb]
\centering
    \psfrag{A}[][][1]{$Q=1$}
    \psfrag{B}[][][1]{$Q=2$}
    \psfrag{C}[][][1]{$Q=3$}
    \psfrag{Z}[][][0.85]{$y=0$}
    \psfrag{D}[][][0.85]{$y=-1$}
    \psfrag{E}[][][0.85]{$y=1$}
    \psfrag{F}[][][0.85]{$y=?$}
    \includegraphics[scale = 0.5]{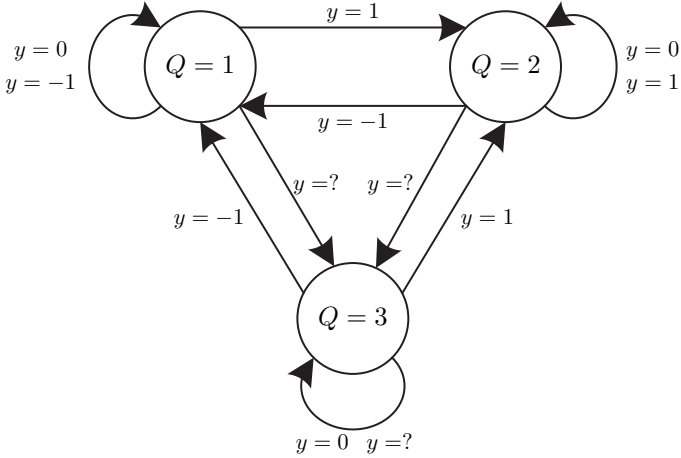}
    \caption{A $Q$-graph for the DEC. This $Q$-graph has a nice interpretation: the nodes $Q=1$ and $Q=2$ correspond to perfect knowledge of the channel state at the decoder, while $Q=3$ is associated with the decoder not knowing the channel state.}
    \label{fig:DEC_Q}
\end{figure}

The proof of Theorem~\ref{theorem: DEC_UB} is given in Appendix \ref{app. DEC_ub}. The bound is obtained by analytically solving the Bellman equation (Theorem \ref{Theorem:Bellman}) corresponding to the DP formulation of the bound in Theorem~\ref{main_bound} obtained from a graph-based test distribution defined on the $Q$-graph in Fig.~\ref{fig:DEC_Q}. Surprisingly, although the feedback capacity optimization problem is very different (the optimization is done over input distributions that are conditioned on past outputs), our upper bound coincides with the DEC feedback capacity. Of course, the feedback capacity is always an upper bound on the capacity without feedback, so the dual capacity method does not yield a better upper bound for this channel.

Nonetheless, our approach serves to illustrate another point. Fig.~\ref{fig:DEC} compares the upper bound in Theorem~\ref{theorem: DEC_UB} with those obtained by optimizing over  first- and second-order Markov test distributions. Since the output alphabet of the DEC is of size $4$ ($\cY = \{-1,0,1,?\}$), the Markov $Q$-graphs of order $k=1$ and $k=2$ have $4$ and $16$ nodes, respectively. Thus, the dual capacity bound obtained using the $Q$-graph on $3$ nodes (depicted in Fig.~\ref{fig:DEC_Q}) outperforms that obtained from Markov $Q$-graphs of larger size. Of course, it is possible that higher-order Markov test distributions may yield bounds that improve upon that in Theorem~\ref{theorem: DEC_UB}, but the problem of optimizing over such test distributions is considerably more complex than that of optimizing over test distributions defined on the $Q$-graph in Fig.~\ref{fig:DEC_Q}. Indeed, exploiting the symmetry between the states $Q=1$ and $Q=2$ in the latter $Q$-graph, the optimization problem over test distributions $R_{Y | Q}$ defined on this graph only involves two free parameters.

For the purpose of comparison, we present below a lower bound on the DEC capacity that is obtained by considering first-order Markov input processes.
\begin{figure}[t] 
\centering
        \psfrag{A}[][][1]{Channel parameter - $\epsilon$}
		\psfrag{B}[][][1]{Rate [bits/symbol]}
		\psfrag{C}[][][1]{DEC - Upper and Lower bounds}
    \includegraphics[scale=0.32]{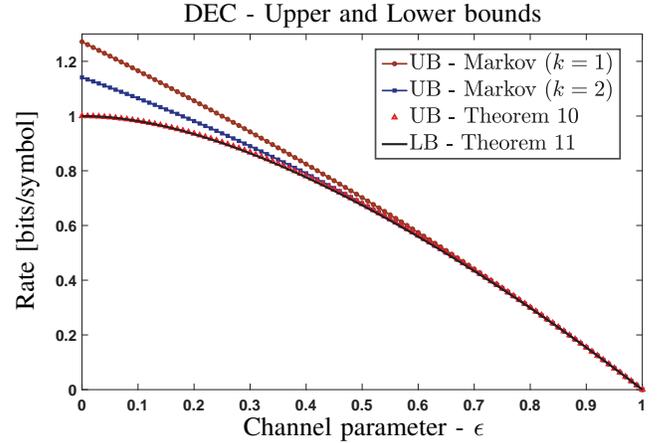} 
    \caption{A comparison between upper and lower bounds on the DEC capacity. The red triangles and the black line represent the upper and the lower bounds from Theorems \ref{theorem: DEC_UB} and \ref{theorem: DEC_LB}, respectively. They almost coincide: the maximum difference between them is $\sim10^{-3}$.}
    \label{fig:DEC}
\end{figure}

\begin{theorem}[\!\cite{henry_dissertation}, Ch. 4] \label{theorem: DEC_LB}
The capacity of the DEC with erasure probability $\epsilon\in[0,1]$ is lower-bounded by the maximum mutual information rate obtained from first-order Markov input processes, which is given by
$$ 
\max_{a\in[0,1]} \ \bar{\epsilon} H_2(a) + \frac{\bar{a}^2\bar{\epsilon}^2}{\epsilon}\sum_{q=0}^{\infty} \left(\frac{\epsilon}{1-a\bar{\epsilon}}\right)^{q+1}H_2\left(\frac{1-(2a-1)^q}{2}\right).
$$
\end{theorem}
The lower bound above is not explicitly stated in \cite[Ch. $4$]{henry_dissertation}, but it can be inferred from the derivations there. For completeness, an alternative proof of this result appears in Appendix~\ref{app. DEC_LB}.

%

\subsection{The POST Channel}
The POST channel was introduced in \cite{POSTchannel} as an example of a channel whose previous output serves as the next channel state. The channel inputs and outputs are related as follows. At each time $t$, if $x_t = y_{t-1}$, then $y_t = x_t$, otherwise, $y_t = x_t\oplus z_t$, where $z_t$ is distributed according to $\mathsf{Bern}(p)$. Accordingly, as illustrated in Fig. \ref{fig:POST}, when $y_{t-1} = 0$, the channel behaves as a $Z$ channel with parameter $p\in[0,1]$, and when $y_{t-1} = 1$, it behaves as an $S$ channel with the same parameter $p$. Here, the new channel state is the channel output, i.e., $s_t = y_t$ and therefore, the POST channel is a unifilar FSC. 

The capacity of the POST channel was found in \cite{POSTchannel}. Here we give an alternative proof of the converse, i.e., an upper bound matching the capacity expression given in \cite{POSTchannel}.

\begin{theorem} \label{Theorem: Post}
The capacity of the POST channel is upper-bounded by
\begin{align*}
\mathsf{C_{POST}} \leq \log_2\left(1+\bar{p}p^{p/\bar{p}}\right),
\end{align*}
for all values of the channel parameter $p\in[0,1]$.
\end{theorem}
The bound is proved by solving the DP formulation of the upper bound in Theorem~\ref{main_bound} obtained from a graph-based test distribution defined on the Markov $Q$-graph depicted in Fig.~\ref{fig:1Markov}. The proof is given in Appendix~\ref{app. POST}. 

\begin{figure}[t]
\centering
\vspace{2mm} 
    \psfrag{A}[][][1]{$0$}
    \psfrag{B}[][][1]{$1$}
    \psfrag{C}[][][1]{$1-p$}
    \psfrag{D}[][][1]{$X_t$}
    \psfrag{E}[][][1]{$Y_t$}
    \psfrag{F}[][][1.1]{$y_{t-1}=0$}
    \psfrag{G}[][][1.1]{$y_{t-1}=1$}
    \psfrag{H}[][][1]{$p$}
\includegraphics[scale = 0.8]{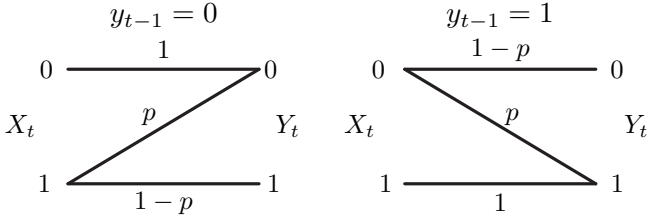}
\caption{POST channel: if $y_{t-1} = 0$ then the channel behaves as a $Z$ channel with parameter $p$, and if $y_{t-1} = 1$ it behaves as an $S$ channel with the same parameter $p$.}
\label{fig:POST}
\end{figure}

\section{Conclusions}\label{sec:conclusion}
In this paper, upper bounds on the capacity of FSCs are derived. First we used the dual capacity bounding technique with graph-based test distributions to derive a multi-letter upper bound expression on the capacity. Then it was shown that the derived upper bound can be formulated as a DP problem, and therefore, the bound is computable. For several channels, we were able to solve explicitly the DP problem, and we presented several results, including novel upper bounds on the capacity of the trapdoor and Ising channels. Further, our results for the DEC demonstrate the value of introducing $Q$-graphs and the accompanying graph-based distributions as a generalization of Markov distributions. An interesting future research direction is to address the complexity of finding a \emph{good} $Q$-graph using efficient reinforcement learning tools to evaluate the dynamic program \cite{AharoniSabagRL}. Such an integrated approach to computing upper bounds is not limited to the channels studied in this paper, and should work for any channel that admits a duality bound, e.g., for channels with feedback \cite{Sabag_DB_FB} and channels with constrained inputs \cite{Dual_Andrew_J}.

\begin{appendices}
\section{Derivation of the upper bound --- Proof of Theorem \ref{main_bound}} \label{app: UB_uniflar}
We provide here a complete proof of the upper bound in Theorem \ref{main_bound}. We start with the relative entropy term in the bound on $I(X^n;Y^n)$ in \eqref{co.DB}: For any initial pair $(s_0,q_0)$, consider
\begin{align}\label{eq: ub1_proof}
	&D\left(P_{Y^n|X^n=x^n,s_0}\middle\|R_{Y^n|Q_0=q_0}\right) \nn\\
	&=  \sum_{y^n} P(y^n|x^n,s_0)\log_2\left(\frac{P(y^n|x^n,s_0)}{R(y^n|q_0)}\right) \nn \\
	&= \sum_{y^n} P(y^n|x^n,s_0) \sum_{j=1}^n\log_2\left(\frac{P(y_j|x^n,y^{j-1},s_0)}{R(y_j|y^{j-1},q_0)}\right)\nn\\
	&\stackrel{(a)}= \sum_{j=1}^n \sum_{y^{j-1}} P(y^{j-1}|x^{j-1},s_0) \sum_{y_j}P(y_j|x^j,y^{j-1},s_0) \nn\\&\times\log_2\left(\frac{P(y_j|x^j,y^{j-1},s_0)}{R(y_j|y^{j-1},q_0)}\right)\nn\\
	&\stackrel{(b)}= \sum_{j=1}^n \sum_{y^{j-1}} P(y^{j-1}|x^{j-1},s_0) \nn\\&\times\infdiv[\big]{P_{Y_j|X^j,Y^{j-1},S_0}(\cdot|x^j,y^{j-1},s_0)}{R(Y_j|y^{j-1},q_0)}\nn\\
    &\stackrel{(c)}= \sum_{j=1}^n \sum_{q_{j-1},s_{j-1}}\sum_{y^{j-1}} P(q_{j-1},s_{j-1},y^{j-1}|x^{j-1},s_0,q_0) \nn\\  &\times \infdiv[\big]{P_{Y_j|X^j,Y^{j-1},S_0}(\cdot|x^j,y^{j-1},s_0)}{R(Y_j|y^{j-1},q_0)}\nn\\
    &\stackrel{(d)}= \sum_{j=1}^n \sum_{q_{j-1},s_{j-1}}\sum_{y^{j-1}} P(q_{j-1},s_{j-1},y^{j-1}|x^{j-1},s_0,q_0)   \nn\\&\times\infdiv[\big]{P_{Y|X,S}(\cdot|x_j,s_{j-1})}{R_{Y|Q}(\cdot|q_{j-1})}\nn\\
	&\stackrel{(e)}= \sum_{j=1}^n \sum_{q,s}z_{j-1}(q,s) \cdot \infdiv[\big]{P_{Y|X,S}(\cdot|x_j,s)}{R_{Y|Q}(\cdot|q)},
\end{align}
where $(a)$ follows by exchanging the order of summation and computing the marginal distributions, $(b)$ follows by identifying the relative entropy, $(c)$  follows from the fact that the pair $(q_{j-1},s_{j-1})$ is a deterministic function of $(x^{j-1},y^{j-1},s_0,q_0)$, $(d)$ follows from the unifilar property and the fact that $q_{j-1}$ is a deterministic function of $(y^{j-1},q_0)$ and, finally, $(e)$ follows since the divergence does not depend on $y^{j-1}$.

By taking the maximum over $x^n$ and dividing the term in \eqref{eq: ub1_proof} by $n$, we obtain, by way of \eqref{co.DB}, 
\begin{align} \label{proof:ub1}
	\mathsf{C} \leq
	\lim_{n\to\infty}\max_{x^n}&\frac{1}{n}\sum_{j=1}^n \sum_{q,s}z_{j-1}(q,s)\nn\\&\times\infdiv[\big]{P_{Y|X,S}(\cdot|x_j,s)}{R_{Y|Q}(\cdot|q)}, 
\end{align}
where the existence of the limit, for any $(s_0, q_0)$, is shown next.

Let us define the quantity 
\begin{align}\label{c_quant}
&\mathsf{c}(x^n,s_0,q_0) \nn\\ &\triangleq \frac{1}{n}\sum\limits_{j=1}^n \sum\limits_{q,s}z_{j-1}(q,s)\infdiv[\big]{P_{Y|X,S}(\cdot|x_j,s)}{R_{Y|Q}(\cdot|q)}.
\end{align}
We will argue that $\lim\limits_{n \to \infty} \max\limits_{x^n} \mathsf{c}(x^n,s_0,q_0)$ exists for any $(s_0,q_0)$, and, in fact, the limit does not depend on the particular choice of $(s_0,q_0)$. To this end, define $\underline{\mathsf{C}}_n$ and $\overline{\mathsf{C}}_n$ to be $\max_{x^n} \min_{s_0,q_0} \mathsf{c}(x^n,s_0,q_0)$ and $\max_{x^n} \max_{s_0,q_0} \mathsf{c}(x^n,s_0,q_0)$, respectively, i.e.,
\begin{align}
\underline{\mathsf{C}}_n \triangleq \frac{1}{n}\max_{x^n\in\mathcal{X}^n}&\min_{s_0,q_0}\sum_{j=1}^n\sum_{q,s}z_{j-1}(q,s) \nn\\&\times\infdiv[\big]{P_{Y|X,S}(\cdot|x_t,s)}{R_{Y|Q}(\cdot|q)}, \label{def:Cn_under} \\
\overline{\mathsf{C}}_n \triangleq \frac{1}{n}\max_{x^n\in\mathcal{X}^n}&\max_{s_0,q_0}\sum_{j=1}^n\sum_{q,s}z_{j-1}(q,s) \nn\\&\times\infdiv[\big]{P_{Y|X,S}(\cdot|x_t,s)}{R_{Y|Q}(\cdot|q)}. \label{def:Cn_over}
\end{align}
For any fixed choice of $(s_0,q_0)$, we clearly have $\underline{\mathsf{C}}_n \le \max_{x^n} \mathsf{c}(x^n,s_0,q_0) \le \overline{\mathsf{C}}_n$. In Appendix~\ref{app:subsec_limexists}, we show that $\lim\limits_{n \to \infty} \underline{\mathsf{C}}_n$ exists, and in Appendix~\ref{app:subsec_initialstate}, we show that this limit in fact equals $\lim\limits_{n \to \infty} \overline{\mathsf{C}}_n$. The desired conclusion follows by a sandwich argument.

\subsection{Existence of $\lim_n \underline{\mathsf{C}}_n$}\label{app:subsec_limexists}
We want to show that $\lim\limits_{n\to\infty} \underline{\mathsf{C}}_n$ exists. The basic idea of the proof is to show that the sequence $n\underline{\mathsf{C}}_n$ is super-additive. A sequence is super-additive if, for any two positive integers $m$, $k$, it satisfies the inequality $a_{m+k}\geq a_m+a_k$. By Fekete's lemma \cite{Fekete1923}, for such a sequence, the limit $\lim\limits_{n\to\infty} \frac{a_n}{n}$ exists, and is equal to $\sup_n \frac{a_n}{n}$.

Let $m$ and $k$ be two positive integers such that $m+k = n$. Let $\hat{x}^m$ and $\hat{x}^k$ be the input sequences that achieve the maximum for $\underline{\mathsf{C}}_m$ and $\underline{\mathsf{C}}_k$, respectively. Now, let $\hat{x}^n$ be the concatenation of $\hat{x}^m$ and $\hat{x}^k$, and consider $x^n=\hat{x}^n$. Since, in general, $x^n$ is not necessarily the input sequence that achieves $n\underline{\mathsf{C}}_n$, we have
\begin{align}\label{limit2}
    &n\underline{\mathsf{C}}_n \nn\\&\geq \min_{s_0,q_0}\sum_{t=1}^n\sum_{q,s}z_{t-1}(q,s) \infdiv[\big]{P_{Y|X,S}(\cdot|x_t,s)}{R_{Y|Q}(\cdot|q)} \nn\\
    &\stackrel{(a)}\geq \min_{s_0,q_0}\sum_{t=1}^m\sum_{q,s}z_{t-1}(q,s) \infdiv[\big]{P_{Y|X,S}(\cdot|x_t,s)}{R_{Y|Q}(\cdot|q)} \nn\\& + \min_{s_0,q_0}\sum_{t=m+1}^n\sum_{q,s}z_{t-1}(q,s) \infdiv[\big]{P_{Y|X,S}(\cdot|x_t,s)}{R_{Y|Q}(\cdot|q)} \nn \\
    &= m\underline{\mathsf{C}}_m \nn\\&+ \min_{s_0,q_0}\sum_{t=m+1}^{m+k}\sum_{q,s}z_{t-1}(q,s) \infdiv[\big]{P_{Y|X,S}(\cdot|x_t,s)}{R_{Y|Q}(\cdot|q)},
\end{align}
where $(a)$ follows from $\min_t\left[f(t)+g(t)\right]\geq \min_t f(t)+\min_t g(t)$. We will now show that the second term in \eqref{limit2} is at least $k\underline{\mathsf{C}}_k$. That is,
\begin{align}
    &\sum_{t=m+1}^{m+k}\sum_{q,s}z_{t-1}(q,s) \infdiv[\big]{P_{Y|X,S}(\cdot|x_t,s)}{R_{Y|Q}(\cdot|q)} \nn \\
    &=\sum_{t=m+1}^{m+k}\sum_{q_{t-1},s_{t-1}}P(q_{t-1},s_{t-1}|x^{t-1},s_0,q_0) \nn\\&\qquad\times\infdiv[\big]{P_{Y|X,S}(\cdot|x_t,s_{t-1})}{R_{Y|Q}(\cdot|q_{t-1})}\nn \\
    & = \sum_{t=m+1}^{m+k}\sum_{q_{t-1},s_{t-1}}\sum_{q_m,s_m}P(q_m,s_m,q_{t-1},s_{t-1}|x^{t-1},s_0,q_0) \nn\\&\qquad\times\infdiv[\big]{P_{Y|X,S}(\cdot|x_t,s_{t-1})}{R_{Y|Q}(\cdot|q_{t-1})} \nn \\
    & \stackrel{(a)}= \sum_{t=m+1}^{m+k}\sum_{q_{t-1},s_{t-1}}\sum_{q_m,s_m}P(q_m,s_m|x^{t-1},s_0,q_0)\nn\\&\qquad\times P(q_{t-1},s_{t-1}|x_{m+1}^{t-1},q_m,s_m) \nn\\&\qquad\times \infdiv[\big]{P_{Y|X,S}(\cdot|x_t,s_{t-1})}{R_{Y|Q}(\cdot|q_{t-1})} \nn \\
    & \stackrel{(b)}= \sum_{t=m+1}^{m+k}\sum_{q_{t-1},s_{t-1}}\sum_{q_m,s_m}P(q_m,s_m|x^m,s_0,q_0)\nn\\&\qquad\times P(q_{t-1},s_{t-1}|x_{m+1}^{t-1},q_m,s_m) \nn\\&\qquad\times \infdiv[\big]{P_{Y|X,S}(\cdot|x_t,s_{t-1})}{R_{Y|Q}(\cdot|q_{t-1})} \nn \\
    & = \sum_{q_m,s_m}P(q_m,s_m|x^m,s_0,q_0)\nn\\&\qquad\times\sum_{t=m+1}^{m+k}\sum_{q_{t-1},s_{t-1}} P(q_{t-1},s_{t-1}|x_{m+1}^{t-1},q_m,s_m) \nn\\&\qquad\times \infdiv[\big]{P_{Y|X,S}(\cdot|x_t,s_{t-1})}{R_{Y|Q}(\cdot|q_{t-1})} \nn \\
    & \ge \min_{q_m,s_m}\sum_{t=m+1}^{m+k}\sum_{q_{t-1},s_{t-1}} P(q_{t-1},s_{t-1}|x_{m+1}^{t-1},q_m,s_m) \nn\\&\qquad\times\infdiv[\big]{P_{Y|X,S}(\cdot|x_t,s_{t-1})}{R_{Y|Q}(\cdot|q_{t-1})} \nn \\
    & =  k\underline{\mathsf{C}}_k, \nn
\end{align}
where $(a)$ follows from the the Markov chain $(Q_{t-1},S_{t-1})-(X_{m+1}^{t-1},Q_m,S_m)-(X^m,S_0,Q_0)$ (see Lemma \ref{lemma: MC1} in Appendix \ref{app:MarkovChain1}), and $(b)$ follows from the Markov chain $(Q_m,S_m)-(X^m,S_0,Q_0)-X_{m+1}^{t-1}$ (see Lemma \ref{lemma: MC1} in Appendix \ref{app:MarkovChain1}). Furthermore, since the minimum over $s_0$ and $q_0$ in \eqref{limit2} does not affect the inequality, we conclude that
\begin{align}
    n\underline{\mathsf{C}}_n &\ge m\underline{\mathsf{C}}_m + k\underline{\mathsf{C}}_k.\nn
\end{align}
Therefore, $n\underline{\mathsf{C}}_n$ is indeed a super-additive sequence, which implies that the limit $\lim_{n}\underline{\mathsf{C}}_n$ exists. 

\subsection{Equality of limits}\label{app:subsec_initialstate}
The following lemma is the main result of this section. 
\begin{lemma}\label{lemma: initial_indep}
If an FSC and a graph-based test distribution, $R_{Y|Q}\succ 0$, are jointly indecomposable, then
    $\lim\limits_{n\to\infty}\underline{\mathsf{C}}_n = \lim\limits_{n\to\infty}\overline{\mathsf{C}}_n$.
\end{lemma}

Before providing the proof of Lemma \ref{lemma: initial_indep}, we present a technical result.
\begin{lemma} \label{lemma: diff} 
Let $Y\in\mathcal{Y}$ and $Z\in\mathcal{Z}$ be two arbitrary random variables such that, for any $z\in\mathcal{Z}$, $P_{Y|Z=z}\ll R_Y$. Then,
\begin{align}
    \left|D\left(P_{Y|Z}\|R_Y|P_Z\right)-D\left(P_Y\|R_Y\right)\right|\le \log_2 (|\mathcal{Z}|).
\end{align}
\end{lemma}
\begin{proof}[Proof of Lemma \ref{lemma: diff}]
We bound the difference as follows:
\begin{align}
    &\left|D\left(P_{Y|Z}\|R_Y|P_Z\right)-D\left(P_Y\|R_Y\right)\right|
    \nn\\&=\left|\sum_{y,z}P(y,z)\log_2\left(\frac{P(y|z)}{R(y)}\right)-\sum_{y}P(y)\log_2\left(\frac{P(y)}{R(y)}\right)\right|\nn\\
    &=\left|-H(Y|Z) + H(Y)\right|\nn\\
    &=I(Y;Z)\nn\\
    &\le \log_2\left(|\mathcal{Z}|\right).
\end{align}
\end{proof}

\begin{proof}[Proof of Lemma \ref{lemma: initial_indep}]
This proof follows the main idea of Gallager's proof in \cite[Theorem $4.6.4$]{Gallager68}. 

From \eqref{eq: ub1_proof}, \eqref{def:Cn_under} and \eqref{def:Cn_over}, we note that for any $n$,
\begin{align*}
    \overline{\mathsf{C}}_n = \max_{x^n\in\mathcal{X}^n}\max_{s_0,q_0}\frac{1}{n}\infdiv[\big]{P_{Y^n|X^n=x^n,s_0}}{R_{Y^n|q_0}}, \\
    \underline{\mathsf{C}}_n = \max_{x^n\in\mathcal{X}^n}\min_{s_0,q_0}\frac{1}{n}\infdiv[\big]{P_{Y^n|X^n=x^n,s_0}}{R_{Y^n|q_0}}.
\end{align*}
Let $x^n$ and $(s_0,q_0)$ be the input sequence, the initial state, and the initial node that maximize $\infdiv[\big]{P_{Y^n|X^n=x^n,s_0}}{R_{Y^n|q_0}}$. Let $(\tilde{s}_0,\tilde{q}_0)$ denote the initial state and the initial node that minimize it for the input sequence $x^n$. Therefore, by the definition of $\overline{\mathsf{C}}_n$ and $\underline{\mathsf{C}}_n$ it follows that
\begin{align}
    \overline{\mathsf{C}}_n &= \frac{1}{n}\infdiv[\big]{P_{Y^n|X^n=x^n,s_0}}{R_{Y^n|q_0}}, \nn\\
    \underline{\mathsf{C}}_n &\ge \frac{1}{n}\infdiv[\big]{P_{Y^n|X^n=x^n,\tilde{s}_0}}{R_{Y^n|\tilde{q}_0}}.\nn
\end{align}
Let $m$ and $k$ be two positive integers such that $m+k = n$. By using the chain rule for relative entropy we have
\begin{align} \label{eq: ub_terms}
    \overline{\mathsf{C}}_n &= \frac{1}{n}\bigg[\infdiv[\big]{P_{Y^k|x^n,s_0}}{R_{Y^k|q_0}} \nn\\&+ \infdiv[\big]{P_{Y_{k+1}^n|Y^k,x^n,s_0}}{R_{Y_{k+1}^n|Y^k,q_0}\big{|}P_{Y^k|x^n,s_0,q_0}}\bigg].
\end{align}


Now, the condition $R_{Y|Q}\succ 0$ assumed in the statement of the lemma ensures that $P_{Y_i|Y^{i-1}=y^{i-1},X^i=x^i}  \ll R_{Y_i|Q_{i-1}=q_{i-1}}$ for any $x^i,y^{i-1}$, where $q_{i-1}=\Phi(y^{i-1})$. There is thus a constant $M_1$ such that
\begin{align*}
    \infdiv[\big]{P_{Y_i|Y^{i-1},x^n,s_0}}{R_{Y_i|Y^{i-1},q_0}\big{|}P_{Y^{i-1}|x^n,s_0,q_0}} \le M_1
\end{align*}
for any $1\le i\le n$. As a consequence, the first relative entropy term in \eqref{eq: ub_terms} is bounded by $kM_1$:
\begin{align} \label{eq: KL_maxx}
    &\infdiv[\big]{P_{Y^k|x^n,s_0}}{R_{Y^k|s_0,q_0}} \nn\\&= \sum_{i=1}^k \infdiv[\big]{P_{Y_i|Y^{i-1},x^n,s_0}}{R_{Y_i|Y^{i-1},q_0}\big{|}P_{Y^{i-1}|x^n,s_0,q_0}} \nn\\
    &\le kM_1. 
\end{align}

Furthermore, by Lemma \ref{lemma: diff}, the second relative entropy term in \eqref{eq: ub_terms} is changed by at most $\log_2(|\mathcal{S}|)$ when conditioning on $S_k$. Therefore, 
\begin{align*}
    \overline{\mathsf{C}}_n &\le \frac{1}{n}\bigg[kM_1 + \log_2(|\mathcal{S}|)\nn\\&+\infdiv[\Big]{P_{Y_{k+1}^n|Y^k,S_k,x^n,s_0}}{R_{Y_{k+1}^n|Y^k,q_0}\Big{|}P_{Y^k,S_k|x^n,s_0,q_0}}\bigg].
\end{align*}
In a similar manner, $\underline{\mathsf{C}}_n$ can be written as in \eqref{eq: ub_terms}, where this time we consider $(\tilde{s}_0,\tilde{q}_0)$ instead of $(s_0,q_0)$. The first term is lower bounded by $0$, and here too, by Lemma \ref{lemma: diff}, we have
\begin{align*} 
    \underline{\mathsf{C}}_n &\ge \frac{1}{n}\bigg[-\log_2(|\mathcal{S}|)\nn\\&+\infdiv[\Big]{P_{Y_{k+1}^n|Y^k,S_k,x^n,\tilde{s}_0}}{R_{Y_{k+1}^n|Y^k,\tilde{q}_0}\Big{|}P_{Y^k,S_k|x^n,\tilde{s}_0,\tilde{q}_0}}\bigg].
\end{align*}

Therefore,
\begin{align}
    &\overline{\mathsf{C}}_n - \underline{\mathsf{C}}_n \nn\\&\le \frac{1}{n}\bigg[kM_1 + 2\log_2(|\mathcal{S}|)\nn\\&+\infdiv[\Big]{P_{Y_{k+1}^n|Y^k,S_k,x^n,s_0}}{R_{Y_{k+1}^n|Y^k,q_0}\Big{|}P_{Y^k,S_k|x^n,s_0,q_0}}\nn\\
     & -\infdiv[\Big]{P_{Y_{k+1}^n|Y^k,S_k,x^n,\tilde{s}_0}}{R_{Y_{k+1}^n|Y^k,\tilde{q}_0}\Big{|}P_{Y^k,S_k|x^n,\tilde{s}_0,\tilde{q}_0}}\bigg] \nn\\ 
    &\stackrel{(a)}= \frac{1}{n} \bigg[kM_1 + 2\log_2(|\mathcal{S}|) 
    \nn\\& + D\Big(P_{Y_{k+1}^n|Y^k,S_k,x^n,s_0}\Big{\|}R_{Y_{k+1}^n|Y^k,Q_k,q_0}\Big{|}P_{Y^k,S_k,Q_k|x^n,s_0,q_0}\Big)  
    \nn\\& - D\Big(P_{Y_{k+1}^n|Y^k,S_k,x^n,\tilde{s}_0}\Big{\|}R_{Y_{k+1}^n|Y^k,Q_k,\tilde{q}_0}\Big{|}P_{Y^k,S_k,Q_k|x^n,\tilde{s}_0,\tilde{q}_0}\Big)\hspace{-0.05cm}\bigg] \nn\\
    &\stackrel{(b)}= \frac{1}{n} \bigg[kM_1 + 2\log_2(|\mathcal{S}|) \nn\\&+ \infdiv[\Big]{P_{Y_{k+1}^n|S_k,x^n,s_0}}{R_{Y_{k+1}^n|Q_k,q_0}\Big{|}P_{S_k,Q_k|x^n,s_0,q_0}}\nn\\ &- \infdiv[\Big]{P_{Y_{k+1}^n|S_k,x^n,\tilde{s}_0}}{R_{Y_{k+1}^n|Q_k,\tilde{q}_0}\Big{|}P_{S_k,Q_k|x^n,\tilde{s}_0,\tilde{q}_0}}\bigg] \nn\\
    &= \frac{1}{n} \bigg[kM_1 + 2\log_2(|\mathcal{S}|)
    \nn\\&+ \sum_{s_k,q_k}\big[P(s_k,q_k|x^k,s_0,q_0)-P(s_k,q_k|x^k,\tilde{s}_0,\tilde{q}_0)\big] \nn
    \\& \times \infdiv[\Big]{P_{Y_{k+1}^n|s_k,x^n}}{R_{Y_{k+1}^n|q_k}}\bigg],
\end{align}
where $(a)$ follows since $Q_k = \Phi(Y^k)$ and $(b)$ follows by observing that the conditioning on $Y^k$ can be dropped due to the channel definition, and since $R(y_{k+1}^n|y^k,q_k,q_0) = R(y_{k+1}^n|q_k,q_0)$.

Here, too, there exists a finite integer $M_2$ such that
\begin{align*}
    \infdiv[\Big]{P_{Y_{k+1}^n|s_k,x^n}}{R_{Y_{k+1}^n|q_k}} \le (n-k)M_2.
\end{align*}
Now, let us denote
\begin{align*}
    \overline{d}_k \triangleq \max_{x^k}\sum_{s_k,q_k}\big|P(s_k,q_k|x^k,s_0,q_0)-P(s_k,q_k|x^k,\tilde{s}_0,\tilde{q}_0)\big|.
\end{align*}
Therefore,
\begin{align*}
    \overline{\mathsf{C}}_n - \underline{\mathsf{C}}_n &\le \frac{1}{n} \bigg[kM_1 + 2\log_2(|\mathcal{S}|) + \overline{d}_k(n-k)M_2\bigg].
\end{align*}

To further upper bound this, we will use Lemma \ref{dist_0} which implies that $\overline{d}_k$ tends to zero as $k$ grows. 
Accordingly, by Lemma \ref{dist_0}, for any $\epsilon > 0$, we can choose $k$ so that $\overline{d}_k\le \epsilon$. Therefore, for such a $k$,
\begin{align*}
    \lim_{n\to\infty}\overline{\mathsf{C}}_n - \underline{\mathsf{C}}_n &\le \epsilon M_2.
\end{align*}
Since $\epsilon > 0$ is arbitrary and $\overline{\mathsf{C}}_n \ge \underline{\mathsf{C}}_n$, the proof is completed.
\end{proof}

\begin{lemma}\label{dist_0}
Consider an FSC and a graph-based test distribution that are jointly indecomposable. Then, for any $\epsilon>0$, there exists an $N$, such that for $n\ge N$
\begin{align} \label{eq: initial_un}
    \big|P(s_n,q_n|x^n,s_0,q_0)-P(s_n,q_n|x^n,\tilde{s}_0,\tilde{q}_0)\big|\le \epsilon
\end{align}
for all $s^n$, $q_n$, $x^n$, $\tilde{s}_0$, $\tilde{q}_0$, $s_0$, and $q_0$.
\end{lemma}

\begin{proof}[Proof of Lemma \ref{dist_0}]
Since the FSC and the graph-based test distribution are jointly indecomposable, then, by Definition \ref{def:sq_ind}, for some fixed $n$ and each input sequence $x^n$, there exists a choice of $s_n$ and $q_n$, such that
\begin{align} \label{eq: suff_cond}
    P(s_n,q_n|x^n,s_0,q_0) > 0, \text{\;\; for all $s_0,q_0.$}
\end{align}
In \cite[Theorem $4.6.3$]{Gallager68}, Gallager provides a sufficient condition for verifying that an FSC is indecomposable, that is, a sufficient condition for verifying that property \eqref{eq: indec: property} holds. 
Following his proof with an appropriate modification we obtain that condition \eqref{eq: suff_cond} is sufficient for verifying that condition \eqref{eq: initial_un} holds. In particular, the modification is done by replacing the state $s_n$ by the pair $(s_n,q_n)$, and the initial state $s_0$ by the pair $(s_0,q_0)$.
\end{proof}
\subsection{Proof of the Markov chains}\label{app:MarkovChain1}
We now show the Markov chains that were required in the proof.

\begin{lemma}\label{lemma: MC1}
For any FSC, the following Markov chains hold:
\begin{align*}
    P(s_t,q_t|x^t,q_m,s_m,s_0,q_0) &= P(s_t,q_t|x_{m+1}^t,q_m,s_m),\\
P(s_m,q_m|x^{t-1},s_0,q_0) &= P(s_m,q_m|x^m,s_0,q_0),
\end{align*}
 for $t\ge m+1$.
\end{lemma}
\begin{proof}[Proof of Lemma \ref{lemma: MC1}]
For the first Markov chain, consider
\begin{align}\label{eq: relation2}
    &P(s_t,q_t|x^t,q_m,s_m,s_0,q_0) \nn\\&= \sum_{y_{m+1}^t}\sum_{s_{m+1}^{t-1}}P(y_{m+1}^t,s_{m+1}^t,q_t|x^t,q_m,s_m,s_0,q_0)\nn\\
        &\stackrel{(a)}= \sum_{y_{m+1}^t}\sum_{s_{m+1}^{t-1}}P(y_{m+1}^t,s_{m+1}^t|x^t,q_m,s_m,s_0,q_0)\nn\\&\qquad\times P(q_t|q_m,y_{m+1}^t),
\end{align}
where $(a)$ follows since $q_t$ is determined by a deterministic function of $q_m$ and the output sequence $y_{m+1}^t$. Further,
\begin{align} \label{eq: Markov2_channel}
    &P(y_{m+1}^t,s_{m+1}^t|x^t,q_m,s_m,s_0,q_0) \nn\\&\stackrel{(a)}= \prod_{i=m+1}^t P(y_i,s_i|y_{m+1}^{i-1},s_{m+1}^{i-1},x^t,q_m,s_m,s_0,q_0) \nn\\
    &\stackrel{(b)}= \prod_{i=m+1}^t P(y_i,s_i|x_i,s_{i-1}),
\end{align}
where $(a)$ follows by the chain rule, and $(b)$ follows by the definition of an FSC.
From \eqref{eq: Markov2_channel}, we observe that $P(y_{m+1}^t,s_{m+1}^t|x^t,q_m,s_m,s_0,q_0)$ does not depend on $x^m$, $s_0$, $q_0$, and therefore, from \eqref{eq: relation2}, so does $P(s_t,q_t|x^t,q_m,s_m,s_0,q_0)$.

For the second Markov chain, consider
\begin{align}\label{eq: relation1}
    &P(s_m,q_m|x^{t-1},s_0,q_0) \nn\\&= \sum_{y^m}\sum_{s^{m-1}}P(y^m,s^m,q_m|x^{t-1},s_0,q_0) \nn\\
        &= \sum_{y^m}\sum_{s^{m-1}} P(y^m,s^m|x^{t-1},s_0,q_0)\mathbbm{1}\{\Phi_{q_0}(y^m)=q_m\}.
\end{align}
Further,
\begin{align}\label{eq: Markov_channel}
    P(y^m,s^m|x^{t-1},s_0,q_0) &\stackrel{(a)}= \prod_{i=1}^m P(y_i,s_i|y^{i-1},s^{i-1},x^{t-1},s_0,q_0)\nn\\
    &\stackrel{(b)}= \prod_{i=1}^m P(y_i,s_i|x_i,s_{i-1}),
\end{align}
where $(a)$ follows by the chain rule, and $(b)$ follows by the definition of an FSC.
From \eqref{eq: Markov_channel}, we observe that $P(y^m,s^m|x^{t-1},s_0,q_0)$ does not depend on $x_{m+1}^{t-1}$, and therefore, from \eqref{eq: relation1}, so does $P(y^m,s^m|x^{t-1},s_0,q_0)$.
\end{proof}

\section{Upper bound for the input-driven FSC (Theorem \ref{main_bound2})} \label{proof_main2}
\begin{proof}
The proof is based on the same main steps we used in the proof of Theorem \ref{main_bound}, but here we consider input-driven FSCs. Let us find an expression equivalent to the conditioned version of the relative entropy term in \eqref{co.DB}. For any initial pair $(s_0,q_0)$ we have, 
\begin{align}\label{eq: ub2_proof}
	&D\big{(}P_{Y^n|X^n=x^n,s_0}\big{\|}R_{Y^n|q_0}\big{)}\nn
	\\&=\sum_{y^n}P_{Y^n|X^n,S_0}(y^n|x^n,s_0)\log_2\left(\frac{P_{Y^n|X^n,S_0}(y^n|x^n,s_0)}{R_{Y^n|Q_0}(y^n|q_0)}\right)\nn
		\\&\stackrel{(a)}= \sum_{j=1}^n \sum_{y^{j-1}} P(y^{j-1}|x^{j-1},s_0)\nn\\&\times \infdiv[\big]{P_{Y_j|Y^{j-1},X^j,S_{0}}(\cdot|y^{j-1},x^j,s_{0})}{R(Y_j|y^{j-1},q_0)}\nn	
		\\&\stackrel{(b)}= \sum_{j=1}^n \sum_{q_{j-1}}\sum_{y^{j-1}} P(q_{j-1},y^{j-1}|x^{j-1},s_0,q_0)\nn 
		\\& \times \infdiv[\big]{P_{Y_j|Y^{j-1},X^j,S_{0}}(\cdot|y^{j-1},x^j,s_{0})}{R(Y_j|y^{j-1},q_0)}\nn
		\\&\stackrel{(c)}= \sum_{j=1}^n \sum_{q_{j-1}}\sum_{y^{j-1}} P(q_{j-1},y^{j-1}|x^{j-1},s_0,q_0)\nn 
		\\& \times D\Bigg(\hspace{-0.07cm}\sum_{s_{j-1}}P(s_{j-1}|x^{j-1},s_0) P_{Y|X,S}(\cdot|x_j,s_{j-1})\Bigg{\|}R(Y|q_{j-1})\hspace{-0.07cm}\Bigg)\nn
		\\&\stackrel{(d)}= \sum_{j=1}^n \sum_{q\in\mathcal{Q}}\beta_{j-1}(q)\nn\\&\times \infdiv[\Bigg]{\sum_{s_{j-1}}\gamma_{j-1}(s_{j-1})\cdot P_{Y|X,S}(\cdot|x_j,s_{j-1})}{R_{Y|Q}(\cdot|q)},
\end{align}
where step $(a)$ follows by computing the marginal distributions, exchanging the order of the summations and identifying the relative entropy, step $(b)$ follows from the fact that $q_{j-1}$ is a deterministic function of $(y^{j-1},q_0)$,
step $(c)$ follows by the input-driven FSC law, i.e.,
\begin{align*}
    &P(y_j|y^{j-1},x^j,s_{0}) \nn\\&=\sum_{s_{j-1}}P(s_{j-1}|x^{j-1},s_0) P_{Y|X,S}(y_j|x_j,s_{j-1}),
\end{align*}
and step $(d)$ follows since the divergence does not depend on $y^{j-1}$.

Therefore, for any $(s_0,q_0)$, we conclude that
\begin{align}
	\mathsf{C} 
	&\stackrel{(a)}\leq \lim_{n\to\infty}\max_{x^n}\frac{1}{n}\sum_{j=1}^n \sum_{q\in\mathcal{Q}}z_{j-1}(q)\nn\\&\times \infdiv[\Bigg]{\sum_{s_{j-1}}\gamma_{j-1}(s_{j-1})\cdot P_{Y|X,S}(\cdot|x_j,s_{j-1})}{R_{Y|Q}(\cdot|q)}\nn,
\end{align}
where $(a)$ follows from the dual upper bound for FSCs and \eqref{eq: ub2_proof}, and $z_{j-1}$, $\gamma_{j-1}$ are defined as
\begin{align*}
	z_{j-1}(q) &\triangleq P_{Q_{j-1}|X^{j-1},S_0,Q_0}(q|x^{j-1},s_0,q_0), \nn\\
	\gamma_{j-1}(s_{j-1}) &\triangleq P_{S_{j-1}|X^{j-1},S_0}(s_{j-1}|x^{j-1},s_0). \nn
\end{align*}
The proofs of the limit's existence and of the initial state independence are omitted as they follow from the same steps taken for the unifilar FSC in Appendix \ref{app: UB_uniflar}.
\end{proof}

\section{DP Formulation of the upper bounds} \label{app: formulation}
\subsection{DP formulation for unifilar FSCs (Theorem \ref{theorem: formulation})}\label{app: formulation_unifilar}
In this section, we prove Theorem \ref{theorem: formulation} on the formulation of the upper bound in Theorem \ref{main_bound} as a dynamic program. The proof has three technical parts: the first two parts are there to verify that the DP is well-defined, and the last part is there is order to relate the average reward of the DP and the upper bound in Theorem \ref{main_bound}. These are summarized in the following lemma.

\begin{lemma}\label{lemma: DP1}
\begin{enumerate}
    \item The reward is a time-invariant function of the DP state and action.
    \item The DP state is a deterministic function of the previous DP state and action.
    \item The limit and the maximization in the upper bound in Theorem \ref{main_bound} can be exchanged. Specifically,
    \begin{align} 
        &\lim_{n\to\infty}\max_{x^n\in\mathcal{X}^n}\min_{s_0,q_0} c(x^n,s_0,q_0)\nn
        \\&=\sup_{\{x_i\}_{i=1}^{\infty}}\liminf_{n\to\infty}\min_{s_0,q_0}c(x^n,s_0,q_0),\nn
    \end{align}
    where $c(x^n,s_0,q_0)$ is defined in \eqref{c_quant}.
\end{enumerate}
Since we showed that the upper bound is independent of the initial state, we can conclude from the third item that $C\le \rho^*$.


\end{lemma}
\begin{proof}[Proof of Lemma \ref{lemma: DP1}]
\begin{enumerate}
    \item Recall that the reward function is defined as
    \begin{align} 
	&g(z_{t-1},u_t) \nn\\&\triangleq \sum_{q,s}z_{t-1}(q,s) \infdiv[\big]{P_{Y|X,S}(\cdot|u_t,s)}{R_{Y|Q}(\cdot|q)}.\nn
    \end{align}
    Therefore, for a fixed FSC and test distribution, it can be easily noted that the reward is a function of the previous DP state $z_{t-1}$ and the action $u_t\triangleq x_t$. 
    \item Let us first derive a recursive relation between $z_t$ at the coordinates $(q_t,s_t)$ and the previous DP state $z_{t-1}$:
    \begin{align}\label{Next_state}
	    &z_{t}(q_t,s_t) \nn\\&= P(q_t,s_t|x^t,s_0,q_0)\nn
	    \\&\stackrel{(a)}= \sum_{q_{t-1},s_{t-1}}P(q_{t-1},s_{t-1}|x^{t-1},s_0,q_0)\nn\\&\qquad\times P(q_t,s_t|q_{t-1},s_{t-1},x^t,s_0,q_0)\nn\\
	    &\stackrel{(b)}= \sum_{q_{t-1},s_{t-1}}z_{t-1}(q_{t-1},s_{t-1})\sum_{y_t}P(y_t|x_t,s_{t-1})\nn\\
	    &\qquad\times\mathbbm{1}\{s_t=f(x_t,y_t,s_{t-1})\}\mathbbm{1}\{q_t=\phi(q_{t-1},y_t)\},
    \end{align}
    where $(a)$ follows from the Markov chain $(Q_{t-1},S_{t-1})-(X^{t-1},S_0,Q_0)-X_t$ that is proven in Appendix \ref{app:MarkovChain1}, and $(b)$ follows from the channel and the $Q$-graph definitions.
    From \eqref{Next_state}, it is clear that $z_t$ is a function of $z_{t-1}$ and the action $x_t$.
    \item The main idea is to show the equality by showing the corresponding two inequalities. The first inequality can be shown as follows:
    \begin{align}
        &\lim_{n\to\infty} \max_{x^n\in\mathcal{X}^n}\min_{s_0,q_0}\frac{1}{n}\sum_{t=1}^n\sum_{q,s}z_{t-1}(q,s) \nn\\&\qquad\qquad\times\infdiv[\big]{P_{Y|X,S}(\cdot|x_t,s)}{R_{Y|Q}(\cdot|q)} \nn \\
        &\ \stackrel{(a)}= \sup_n\max_{x^n\in\mathcal{X}^n}\min_{s_0,q_0}\frac{1}{n}\sum_{t=1}^n\sum_{q,s}z_{t-1}(q,s) \nn\\&\qquad\qquad\times\infdiv[\big]{P_{Y|X,S}(\cdot|x_t,s)}{R_{Y|Q}(\cdot|q)} \nn \\
        &\ =\sup_{\{x_i\}_{i=1}^{\infty}}\sup_n\min_{s_0,q_0}\frac{1}{n}\sum_{t=1}^n\sum_{q,s}z_{t-1}(q,s) \nn\\&\qquad\qquad\times\infdiv[\big]{P_{Y|X,S}(\cdot|x_t,s)}{R_{Y|Q}(\cdot|q)} \nn \\
        &\ \ge\sup_{\{x_i\}_{i=1}^{\infty}}\liminf_{n\to\infty}\min_{s_0,q_0}\frac{1}{n}\sum_{t=1}^n\sum_{q,s}z_{t-1}(q,s) \nn\\&\qquad\qquad\times\infdiv[\big]{P_{Y|X,S}(\cdot|x_t,s)}{R_{Y|Q}(\cdot|q)},
    \end{align}
    where $(a)$ follows by Fekete's lemma (see Appendix~\ref{app:subsec_limexists} where it is shown that the sequence $n\underline{\mathsf{C}}_n$ is supper additive).
    
    We now show the reverse inequality. Using the notation and the main result from Appendix~\ref{app:subsec_limexists}, the existence of $\lim\limits_{n\to\infty} \underline{\mathsf{C}}_n$ implies that, for any $\epsilon>0$, there exists an $N(\epsilon)$ such that for all $k>N(\epsilon)$
    \begin{align} \label{eq: n_eps}
        \underline{\mathsf{C}}_k \geq \lim_{n\to\infty} \underline{\mathsf{C}}_n -\epsilon.
    \end{align}
    
    Fix $k>N(\epsilon)$, and let $\hat{x}^k$ be the input sequence that achieves the maximum. Define $\tilde{x}^{\infty}=\{\tilde{x}_t\}_{t= 1}^{\infty}$ as an infinite sequence composed of identical concatenations of the sequence $\hat{x}^k$. Consider the following chain of inequalities
    \begin{align}
        &\sup_{\{x_i\}_{i=1}^{\infty}}\liminf_{n\to\infty}\min_{s_0,q_0}\frac{1}{n}\sum_{t=1}^n\sum_{q,s}z_{t-1}(q,s) \nn\\&\qquad\qquad\times\infdiv[\big]{P_{Y|X,S}(\cdot|x_t,s)}{R_{Y|Q}(\cdot|q)}\nn\\
        &\stackrel{(a)}\ge \liminf_{n\to\infty}\min_{s_0,q_0}\frac{1}{n}\sum_{t=1}^n\sum_{q,s}z_{t-1}(q,s) \nn\\&\qquad\qquad\times\infdiv[\big]{P_{Y|X,S}(\cdot|\tilde{x}_t,s)}{R_{Y|Q}(\cdot|q)} \nn\\
        &\stackrel{(b)}= \liminf_{n\to\infty}\min_{s_0,q_0}\frac{1}{n}\sum_{t=1}^{k\lfloor\frac{n}{k}\rfloor}\sum_{q,s}z_{t-1}(q,s) \nn\\&\qquad\qquad\times\infdiv[\big]{P_{Y|X,S}(\cdot|\tilde{x}_t,s)}{R_{Y|Q}(\cdot|q)} \nn\\
        &\stackrel{(c)}\ge \liminf_{n\to\infty}\frac{k}{n}\left\lfloor \frac{n}{k} \right\rfloor\bigg[\frac{1}{k}\cdot\min_{s_0,q_0}\sum_{t=1}^{k}\sum_{q,s}z_{t-1}(q,s) \nn\\&\qquad\qquad\times\infdiv[\big]{P_{Y|X,S}(\cdot|\tilde{x}_t,s)}{R_{Y|Q}(\cdot|q)}\bigg]\nn\\ 
        &\stackrel{(d)}\ge 
        \lim_{n\to\infty}\underline{\mathsf{C}}_n-\epsilon,
    \end{align}
    where $(a)$ follows by considering the sequence $\tilde{x}^{\infty}$, which is not necessarily the input sequence that achieves the maximum, $(b)$ follows from the fact that $k$ is fixed and the divergence is bounded, and therefore, when rounding $n$ to $k\lfloor n/k\rfloor$ the residual goes to zero, $(c)$ follows from taking the minimum at the beginning of each $k$th block, i.e., $\min_t\sum_if_i(t)\geq \sum_i\min_t f_i(t)$ and, the fact that $\tilde{x}^\infty$ is a repetition of the same sequence $\hat{x}^k$, and, finally, $(d)$ follows from \eqref{eq: n_eps}.
\end{enumerate}  
\end{proof}

\subsection{DP formulation for input-driven FSCs (Theorem \ref{theorem: formulation_input_dependent})}\label{App: formulation_input_dependent}
In this section, we prove Theorem \ref{theorem: formulation_input_dependent} on the formulation of the upper bound in Theorem \ref{main_bound2} as a dynamic program. Similarly to Lemma \ref{lemma: DP1}, the proof consists of three technical parts that are summarized in the following lemma.
\begin{lemma}\label{lemma: DP2}
\begin{enumerate}
    \item The reward is a time-invariant function of the DP state and action.
    \item The DP state is a deterministic function of the previous DP state and action.
    \item The limit and the maximization in the upper bound can be exchanged. Specifically,
    \begin{align} 
        &\lim_{n\to\infty} \max_{x^n\in\mathcal{X}^n}\frac{1}{n}\sum_{t=1}^n\sum_{q\in \cQ}\beta_{t-1}(q) \nn\\&\times D\Bigg(\sum_{s_{t-1}}\gamma_{t-1}(s_{t-1}) P_{Y|X,S}(\cdot|x_t,s_{t-1})\Bigg{\|} R_{Y|Q}(\cdot|q)\Bigg)\nn
        \\&=\sup_{\{x_i\}_{i=1}^{\infty}}\liminf_{n\to\infty}\frac{1}{n}\sum_{t=0}^{n}\sum_{q\in\mathcal{Q}}\beta_{t-1}(q) \nn\\&\times D\Bigg(\sum_{s_{t-1}}\gamma_{t-1}(s_{t-1}) P_{Y|X,S}(\cdot|x_t,s_{t-1})\Bigg{\|} R_{Y|Q}(\cdot|q)\Bigg).\nn
    \end{align}
\end{enumerate}
Here, also, the upper bound is independent of the initial state. Therefore, we can conclude from the third item that $C\le \rho^*$.
\end{lemma}

\begin{proof}[Proof of Lemma \ref{lemma: DP2}]
\begin{enumerate}
    \item The reward function in Eq. \eqref{reward_2} is defined as 
    \begin{align*}
	    &g(z_{t-1},u_t) \triangleq \sum_{q}\beta_{t-1}(q)\nn\\&\times \infdiv[\Bigg]{\sum_{s}\gamma_{t-1}(s)\cdot P_{Y|X,S}(\cdot|u_t,s)}{R_{Y|Q}(\cdot|q)}.
    \end{align*}
    Accordingly, since $z_{t-1} = (\beta_{t-1},\gamma_{t-1})$, this item is deduced directly from the definition above.
    \item Let us first derive a recursive relation between $z_t = (\beta_t,\gamma_t)$ and the previous DP state $z_{t-1}$. In particular, $\beta_t$ is computed as
    \begin{align} \label{DP4_state1}
	    &\beta_{t}(q_t) \nn\\&\triangleq P(q_t|x^t,s_0,q_0) \nn\\
	    &\stackrel{(a)}= \sum_{q_{t-1},s_{t-1}} P(q_{t-1},s_{t-1}|x^{t-1},s_0,q_0)\nn\\&\quad \times P(q_t|q_{t-1},s_{t-1},x^t,s_0,q_0) \nn
	    \\
	    &\stackrel{(b)}= \sum_{q_{t-1}} P(q_{t-1}|x^{t-1},s_0,q_0)\sum_{s_{t-1}}P(s_{t-1}|x^{t-1},s_0)\nn\\&\quad \times P(q_t|q_{t-1},s_{t-1},x^t,s_0,q_0) \nn\\
	    &\stackrel{(c)}= \sum_{q_{t-1}}\beta_{t-1}(q_{t-1})\sum_{s_{t-1}}\gamma_{t-1}(s_{t-1})\nn\\&\quad \times \sum_{y_t} P_{Y|X,S}(y_t|x_t,s_{t-1})\mathbbm{1}\{q_t=\phi(q_{t-1},y_t)\}, 
    \end{align}
    where $(a)$ follows from the Markov chain $(Q_{t-1},S_{t-1})-(X^{t-1},S_0,Q_0)-X_t$ that is proven in Appendix \ref{app:MarkovChain1}, $(b)$ follows from the Markov chain $S_{t-1}-(X^{t-1},S_0)-(Q_{t-1},Q_0)$, and $(c)$ follows from the channel characteristics and the $Q$-graph definition. Furthermore, $\gamma_{t}$ is computed as
    \begin{align} \label{DP4_state2}
        \gamma_{t}(s_t) &= P(s_t|x^t,s_0) \nn\\
        &\stackrel{(a)}=\sum_{s_{t-1}}\gamma_{t-1}(s_{t-1})\;P(s_t|x_t,s_{t-1}),
    \end{align}
    where $(a)$ follows from the Markov chain $S_{t-1}-(X^{t-1},S_0)-X_t$ and the input-driven FSC definition in \eqref{eq:input_def}. From \eqref{DP4_state1} and \eqref{DP4_state2}, it is clear that $\beta_t$ and $\gamma_t$ are a function of the previous DP state $z_{t-1}$ and the action $x_t$.
    \item The proof of this item is omitted as it follows from the same steps taken for unifilar FSCs in Appendix \ref{app: formulation_unifilar}.
\end{enumerate}
\end{proof}

\section{Trapdoor Channel --- Proof of Theorem \ref{theorem: Trapdoor}} \label{app. trapdoor}
\begin{proof}
The proof is based on the Markov $Q$-graph from Fig. \ref{fig:1Markov} and on the following optimized graph-based test distribution:
\begin{align*}
	R_{Y|Q}(0|0) = R_{Y|Q}(1|1) = \frac{2}{3}.
\end{align*}

Since the trapdoor channel is a unifilar FSC, we define $\underline{z}$ as a pmf on $\mathcal{Q} \times \mathcal{S}$ that corresponds to the DP state in Section \ref{DP_main1}. In particular, $\underline{z}$ consist of four elements that are indexed as $z_{q,s}$ where $z_{q,s}=P(q,s)$. To simplify notation, we will consider in the calculation below the relation $z_{1,1} = 1-z_{0,0}-z_{0,1}-z_{1,0}$.

Recall that to solve the Bellman equation (Theorem \ref{Theorem:Bellman}), one should identify a scalar $\rho$ and a function $h:\mathcal{Z}\rightarrow\mathbb{R}$ such that
\begin{align} \label{Bellman_Trapdoor}
\rho+h(\underline{z}) = \max_u \left(g(\underline{z},u) + h\left(F(\underline{z},u)\right)\right).
\end{align}
In the following, we show that $\rho^* = \log_2\left(\frac{3}{2}\right)$ and the function 
\begin{align*}
	h^*(\underline{z})  
	&=\left\{\begin{array}{cc}
        z_{1,0}, & z_{0,1} \le z_{1,0}, \\
        z_{0,1}, & z_{0,1} > z_{1,0}. \end{array}\right.
\end{align*}
solves \eqref{Bellman_Trapdoor}.

The reward function can be computed as 
\begin{align*}
g(\underline{z},u) &= \begin{cases}         \log_2\left(\frac{3}{2}\right)+\frac{1}{2}\left(z_{0,0} + 3z_{1,0}-1\right), & u=0, \\
\log_2\left(\frac{3}{2}\right)+\frac{1}{2}\left(z_{1,1} + 3z_{0,1} -1\right), & u=1.
\end{cases}
\end{align*}
The next DP state, defined in Eq. \eqref{Next_state}, is given by
\begin{align*} \label{policy_Trapdoor}
    &F(\underline{z},u) \nn\\&= \begin{cases} \left[z_{0,0}+z_{1,0},\frac{1}{2}(z_{0,1}+z_{1,1}),\frac{1}{2}(z_{0,1}+z_{1,1})\right], & u=0, \\
    \left[0,\frac{1}{2}(z_{0,0}+z_{1,0}),\frac{1}{2}(z_{0,0}+z_{1,0})\right], & u=1.
    \end{cases}
\end{align*}
Let us assume that the optimal policy $u^*(\underline{z})$ is given by
\begin{align}
    u^*(\underline{z})=
    \begin{cases}
          0, & z_{0,1}\le z_{1,0}, \\
          1, & z_{0,1} > z_{1,0}.
    \end{cases}
\end{align}

Assuming \eqref{policy_Trapdoor}, then for $z_{0,1}\le z_{1,0}$, the left-hand side of the Bellman equation is equal to
\begin{align}
    \rho^*+h^*(\underline{z}) = \log_2\left(\frac{3}{2}\right)+z_{1,0}, \nn
\end{align}
while the right-hand side of the Bellman equation is
\begin{align*}
g(\underline{z},u=0) + h^*\left(F(\underline{z},u=0)\right) 
                  = \log_2\left(\frac{3}{2}\right) + z_{1,0}.
\end{align*}
Hence, assuming \eqref{policy_Trapdoor}, we showed that the Bellman equation is satisfied for $z_{0,1}\le z_{1,0}$. It can also be verified that the Bellman equation is satisfied when $z_{0,1}>z_{1,0}$.

We will now verify that the assumption we made in \eqref{policy_Trapdoor} holds. That is,
\begin{align*}
&\left(g(\underline{z},u=0) + \frac{1}{2}(z_{0,1}+z_{1,1})\right) \nn\\&- \left(g(\underline{z},u=1) + \frac{1}{2}(z_{0,0}+z_{1,0})\right) = z_{1,0}-z_{0,1}, 
\end{align*}
which is nonnegative for all $z_{1,0}\ge z_{0,1}$, and therefore, in this region, $u=0$ is indeed the optimal action.
Similarly, it can also be verified that, for all $z_{1,0} < z_{0,1}$, $u=1$ is the optimal action.

Therefore, we conclude that, $\rho^*=\log_2\left(\frac{3}{2}\right)$ is indeed the optimal average reward.
\end{proof}

\section{Ising Channel --- Proof of Theorem \ref{theorem: Ising}} \label{app. Ising}
\begin{proof}
The proof is based on a Markov $Q$-graph with $k=3$. Recall that for the Ising channel the state is evaluated according to $s_t = x_t$. Therefore, we can use the simplified DP formulation that is presented in Section \ref{S_DP}. The proof of the bound is based on the following graph-based test distribution:
\begin{align*}
	R_{Y|Q}(0|0,0,0) &= 1 - R_{Y|Q}(0|1,1,1) = a \\
	R_{Y|Q}(0|0,1,0) &= 1 - R_{Y|Q}(0|1,0,1) = b \\
	R_{Y|Q}(0|1,0,0) &= 1 - R_{Y|Q}(0|0,1,1) = c \\
	R_{Y|Q}(0|1,1,0) &= 1 - R_{Y|Q}(0|0,0,1) = d,
\end{align*}
where $[a,b,c,d]\in(0,1)^4$. Let $\underline{z} = [z_0,z_1,z_2,z_3]$ denote the DP state vector, where $\{z_i\}_{i=0}^3\in\{0,1\}$.
According to the DP formulation, the next DP state is computed as $F(\underline{z},u) = [z_1,z_2,z_3,u]$, and the reward function is defined as
\begin{align*}
	g(\underline{z},u) =& \sum_{y_1^3} \left(\prod_{i=1}^{3} P_{Y|X,S}(y_i|z_i,z_{i-1})\right)\nn\\ &\times \infdiv[\big]{P_{Y|X,S}(\cdot|u,z_3)}{R_{Y|Q}(\cdot|y_1^3)}.
\end{align*}

According to Theorem \ref{Theorem:Bellman}, if we identify a scalar $\rho$ and a bounded function $h(\underline{z})$ such that
\begin{align} \label{eq:Ising_bellman}
\rho+h(\underline{z}) &= \max_u \left[g(\underline{z},u)+h\left(F(\underline{z},u)\right)\right] ,\;\; \forall \underline{z}\in\mathcal{Z},
\end{align}
then $\rho=\rho^*$. 
In the following, we show that
 \begin{align}
	\rho^* &= \frac{1}{4}\log_2\left(\frac{1}{2acd(1-a)}\right) 
\end{align}
and the function $h^*(\underline{z})$ defined below solves \eqref{eq:Ising_bellman}. 
\begin{align}
	h^*(0,0,0,0) &= h^*(1,1,1,1) = \frac{1}{2}\log_2\left(\frac{1}{4a\bar{a}}\right) \nonumber \\
	h^*(0,0,0,1) &= h^*(1,1,1,0) = \frac{1}{4}\log_2\left(\frac{1}{2acd\bar{a}}\right) \nonumber \\
	h^*(0,0,1,0) &= h^*(1,1,0,1) = \frac{1}{8}\log_2\left(\frac{\bar{a}^3cd}{64ab^5\bar{b}\bar{c}^5\bar{d}^3}\right) \nonumber \\
	h^*(0,0,1,1) &= h^*(1,1,0,0) =\frac{1}{2}\log_2\left(\frac{1}{2ac}\right) \nonumber \\
	h^*(0,1,0,0) &= h^*(1,0,1,1) = \frac{1}{8}\log_2\left(\frac{\bar{a}d}{256ab^3c\bar{b}\bar{c}^3\bar{d}}\right) \nonumber \\
	h^*(0,1,0,1) &= h^*(1,0,1,0) =\frac{1}{4}\log_2\left(\frac{\bar{a}d}{8b^2\bar{b}\bar{c}^2\bar{d}}\right) \nonumber \\
	h^*(0,1,1,0) &= h^*(1,0,0,1) = \frac{1}{4}\log_2\left(\frac{1}{2abc\bar{c}}\right) \nonumber \\
	h^*(0,1,1,1) &= h^*(1,0,0,0) = \frac{1}{4}\log_2\left(\frac{d}{8a^3c\bar{a}}\right).
\end{align}

Let us assume that the optimal policy, under the constraints given in \eqref{eq:Ising_constraints}, is given by
\begin{align} \label{policy_Ising}
    u^*(\underline{z}) = \bar{z_0}\bar{z_2}+z_3\cdot(z_0\oplus z_2),
\end{align}
where $\oplus$ denotes the XOR operation. The policy in \eqref{policy_Ising} is obtained by optimizing the DP program and extracting the relation between the optimal policy and the DP state.
Assuming \eqref{policy_Ising}, it can now be verified that \eqref{eq:Ising_bellman} is satisfied with the above choice of $\rho^*$ and the function $h^*(\underline{z})$. Here, we will verify that it holds only for $\underline{z}=[0,0,0,0]$, and the verification for the other states can be done similarly. The left-hand side of the Bellman equation is
\begin{align}
    \rho^*+h^*(0,0,0,0) = \frac{1}{4}\log_2\left(\frac{1}{32a^3cd(1-a)^3}\right),\nn
\end{align}
while the right-hand side of the Bellman equation is
\begin{align}
&\max_{u} \left[g(0,0,0,0,u)+h(0,0,0,u)\right]\nn\\
    &\stackrel{(a)}= g(0,0,0,0,1)+h(0,0,0,1) \nonumber \\
    &= \frac{1}{2}\log_2\left(\frac{1}{4a(1-a)}\right) + \frac{1}{4}\log_2\left(\frac{1}{2acd(1-a)}\right) \nonumber \\
    &= \frac{1}{4}\log_2\left(\frac{1}{32a^3cd(1-a)^3}\right),\nn
\end{align}
where $(a)$ follows from \eqref{policy_Ising}, and therefore the Bellman equation holds for $\underline{z}=[0,0,0,0]$. It is now left to verify that the suggested policy in \eqref{policy_Ising} is indeed optimal under the constraints given in \eqref{eq:Ising_constraints}. Here, too, we will verify it only for $\underline{z}=[0,0,0,0]$ and the verification for the other states can be done similarly.
\begin{align} \label{eq:Ising_nonnegetive}
    &\left[g(0,0,0,0,1)+h(0,0,0,1)\right]-\left[g(0,0,0,0,0)+h(0,0,0,0)\right] \nn\\ &=\frac{1}{4}\log_2\left(\frac{1}{32a^3cd(1-a)^3}\right)-\frac{1}{4}\log_2\left(\frac{1}{16a^6\bar{a}^2}\right)\nn \\
    &= \frac{1}{4}\log_2\left(\frac{a^3}{2cd\bar{a}}\right),
\end{align}
where we note that \eqref{eq:Ising_nonnegetive} is nonnegative when $0\leq a^3-2\bar{a}cd$. Therefore, under the constraints in \eqref{eq:Ising_constraints}, $u=1$ is indeed the optimal action when $\underline{z}=[0,0,0,0]$.

\end{proof}
\section{DEC --- Proof of Theorem \ref{theorem: DEC_UB}} \label{app. DEC_ub}
\begin{proof}
The proof is based on the $Q$-graph depicted in Fig. \ref{fig:DEC_Q} and on the following graph-based test distribution:
\begin{align*}
	R_{Y|Q}  = \begin{bmatrix} 
0 & 0.5\bar{\epsilon} & 0.5\bar{\epsilon} & \epsilon \\
0.5\bar{\epsilon} & 0.5\bar{\epsilon} & 0 & \epsilon \\
0.5p\bar{\epsilon} & \bar{p}\bar{\epsilon} & 0.5p\bar{\epsilon} & \epsilon \\
\end{bmatrix},
\end{align*}
where $p\in[0,1]$, the rows correspond to $Q=1,2,3$ and the columns correspond to $Y=-1,0,1,?$ in that order.

Now, note that some of the test distribution entries are equal to zero, and therefore, the condition $R_{Y|Q}\succ 0$ in Theorem~\ref{main_bound} does not hold. However, it can be easily verified that the condition in Remark \ref{remark: UB_restriction} holds. This is mainly due to the fact that when $Q=1$ the previous state must be equal to $0$, and when $Q=2$ the previous state must be equal to $1$. We omit the details of this verification.

Using the above choice of a test distribution, one can show that the Bellman equation holds. However, since the upper bound is exactly equal to the feedback capacity, and $\mathsf{C}\le \mathsf{C_{FB}}$ (where $\mathsf{C_{FB}}$ denotes the feedback capacity), we will not provide here the proof that the Bellman equation holds. It will only be shown that the resultant upper bound expression in Theorem \ref{theorem: DEC_UB} is equal to the feedback capacity \cite{Sabag_UB_IT}.

The feedback capacity of the DEC is given by
\begin{align*}
    \mathsf{C_{FB}}= \max_{p\in[0,1]} (1-\epsilon)\frac{p+\epsilon H_2(p)}{\epsilon+(1-\epsilon)p}.
\end{align*}
Denote
\begin{align*}
    G(p,\epsilon) = (1-\epsilon)\frac{p+\epsilon H_2(p)}{\epsilon+(1-\epsilon)p}.
\end{align*}
Straightforward calculations show that the derivative of $G(p,\epsilon)$ (with respect to $p$) is equal to zero iff
\begin{align} \label{DEC_derv_0}
    G(p,\epsilon)= 1+\epsilon\log_2\left(\frac{1-p}{p}\right).
\end{align}
Therefore, $\mathsf{C_{FB}} = G(p^*,\epsilon)$ where $p^* = \argmax_{p\in[0,1]} G(p,\epsilon)$.
Using simple algebra, it can be further verified that \eqref{DEC_derv_0} holds iff $2\bar{p}=p^{\epsilon}$. Hence, $p^*$ is the solution $p$ of the equation $2\bar{p}=p^{\epsilon}$.
\end{proof}

\section{DEC --- Proof of Theorem \ref{theorem: DEC_LB}} \label{app. DEC_LB}
\begin{proof}
The basic idea of the lower bound proof is to consider input sequences that are restricted to a first-order Markov process, i.e.,
\begin{align}\label{eq: Markov1}
    P_{X^n}(x^n) = \prod_{i=1}^n P_{X|X^-}(x_i|x_{i-1}).
\end{align}
In the following we denote by $\cal{P_{\mathrm{markov}}}$ the set of all distributions satisfying \eqref{eq: Markov1}. The capacity of the DEC is then lower bounded by
\begin{align}\label{eq: DEC_Markov_obj}
    \mathsf{C_{DEC}} \ge \lim_{n\to\infty}\max_{P(x^n)\in\cal{P_{\mathrm{markov}}}}\frac{1}{n}I(X^n;Y^n|S_0=s_0)
\end{align}
for any $s_0\in\mathcal{S}$. Based on the channel symmetry, we consider the following input distribution:
\begin{align*}
	P_{X|X^-}(0|0) = P_{X|X^-}(1|1) = a,
\end{align*}
where $a\in[0,1]$.
In the following, we will find the mutual information in \eqref{eq: DEC_Markov_obj} explicitly:
\begin{align} \label{eq: DEC_tocalc}
    &I(X^n;Y^n|S_0=s_0) \nn\\&= H(Y^n|S_0=s_0) - H(Y^n|X^n,S_0=s_0) \nn \\
        &= \sum_{i=1}^n \left[H(Y_i|Y^{i-1},S_0=s_0) - H(Y_i|Y^{i-1},X^n,S_0=s_0)\right] \nn\\
        &\stackrel{(a)}= \sum_{i=1}^n H(Y_i|Y^{i-1},S_0=s_0) - nH_2(\epsilon),
\end{align}
where $(a)$ follows by the Markov chain $Y_i-(X_i,X_{i-1})-(X^{i-2},X_{i+1}^n,Y^{i-1})$ and the channel law. To find $H(Y_i|Y^{i-1},S_0=s_0)$, we will calculate the probabilities $P(y_i|y^{i-1},s_0)$.
First, let us find the distribution $P(x_i=0|y^i,s_0)$ for any possible output sequence $y^i$. We will show that this distribution induces the graph depicted in Fig. \ref{fig:DEC_Q_ext}. 
For any output sequence $y^{i-1}$,
\begin{align}
    P(x_i=0|y_i=-1,y^{i-1},s_0) = 1, \label{eq: DEC_y_s0}\\
    P(x_i=0|y_i=1,y^{i-1},s_0) = 0, \label{eq: DEC_y_s1}
\end{align}
where \eqref{eq: DEC_y_s0} follows since the channel output is $y_i=-1$ iff $x_i=0$, and \eqref{eq: DEC_y_s1} follows since the channel output is $y_i=1$ iff $x_i=1$. Further,
\begin{align}\label{eq: DEC_y0}
    &P(x_i=0|y_i=0,y^{i-1},s_0) \nn\\
    &= \frac{\sum_{x_{i-1}}P(x_{i-1}|y^{i-1},s_0)P(x_i=0|x_{i-1})P_{Y|X,S}(0|0,x_{i-1})}{\sum_{x_{i-1}^i}P(x_{i-1}|y^{i-1},s_0)P(x_i|x_{i-1})P_{Y|X,S}(0|x_i,x_{i-1})}\nn\\
    &= \frac{a\bar{\epsilon}P(x_{i-1}=0|y^{i-1},s_0)}{a\bar{\epsilon}P(x_{i-1}=0|y^{i-1},s_0) + a\bar{\epsilon}P(x_{i-1}=1|y^{i-1},s_0)}\nn\\
    &= P(x_{i-1}=0|y^{i-1},s_0),
\end{align}
and
\begin{align}\label{eq: DEC_y_?}
    &P(x_i=0|y_i=?, y^{i-1},s_0) 
    \nn\\&= \frac{\sum_{x_{i-1}}P(x_{i-1},x_i=0,y_i=?|y^{i-1},s_0)}{\sum_{x_{i-1}^i}P(x_{i-1},x_i,y_i=?|y^{i-1},s_0)}\nn\\
    &=\sum_{x_{i-1}}P(x_{i-1}|y^{i-1},s_0)P(x_i=0|x_{i-1})\nn\\
    &= aP(x_{i-1}=0|y^{i-1},s_0)+\bar{a}P(x_{i-1}=1|y^{i-1},s_0).
\end{align}
\begin{figure}[b!]
\centering
\psfrag{A}[][][0.7]{$Q=A_0$}
    \psfrag{B}[][][0.7]{$Q=B_0$}
    \psfrag{H}[][][0.7]{$Q=A_1$}
    \psfrag{I}[][][0.7]{$Q=B_1$}
    \psfrag{J}[][][0.7]{$Q=A_n$}
    \psfrag{K}[][][0.7]{$Q=B_n$}
    \psfrag{Q}[][][0.9]{$0/1$}
    \psfrag{R}[][][0.9]{$0/-1$}
    \psfrag{Z}[][][0.9]{$0$}
    \psfrag{D}[][][0.9]{$-1$}
    \psfrag{E}[][][0.9]{$1$}
    \psfrag{F}[][][0.9]{$?$}
    \includegraphics[scale = 0.42]{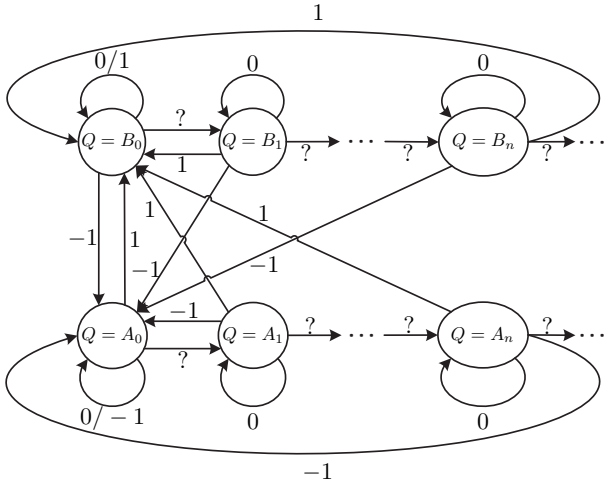}
    \caption{A $Q$-graph with $\mathcal{Y}=\{-1,0,1,?\}$. Each node on the graph has an outgoing edge labeled with $y=-1$ to the node $Q=A_0$ and an outgoing edge labeled with $y=1$ to the node $Q=B_0$. Further, each node has a self-loop labeled with $y=0$.}
    \label{fig:DEC_Q_ext}
\end{figure}
Based on \eqref{eq: DEC_y_s0}--\eqref{eq: DEC_y_?} we now show that the probability $P(x_i=0|y^i,s_0)$ induces the graph depicted in Fig \ref{fig:DEC_Q_ext}. Equations \eqref{eq: DEC_y_s0} and \eqref{eq: DEC_y_s1} imply that, for any possible output sequence $y^{i-1}$, if $y_i=1$ or $y_i=-1$, then $P(x_i=0|y^i)$ is equal to 0 or 1, respectively. Therefore, each node on the graph in Fig \ref{fig:DEC_Q_ext} has an outgoing edge labeled with $y=-1$ to $Q=A_0$ and an outgoing edge labeled with $y=1$ to $Q=B_0$.
Equation \eqref{eq: DEC_y0} implies that each possible node on the graph has a self-loop labeled with $y=0$. Finally, \eqref{eq: DEC_y_?} implies that, if the current output is $y_i=?$, then there is an outgoing edge labeled with $y=?$ to the next node on the graph, as depicted in Fig \ref{fig:DEC_Q_ext}. Note that the induced graph contains an infinite number of nodes.

To conclude, given an initial node $q_0$, there exists a unique mapping $\Phi_{q_0}:{\cY}^{i}\to \cQ$ from an output sequence $y^{i}$ to a unique node on the induced graph. Therefore, the equality $P(x_i=0|y^i,s_0) = P(x_i=0|q_i,s_0)$ holds where $q_i = \Phi(y^i)$.
Accordingly, using \eqref{eq: DEC_y_s0}--\eqref{eq: DEC_y_?}, it follows that for $q\in\mathbb{N}\cup\{0\}$
\begin{align}
    &P(x_i=0|y^i,s_0) \nn\\&= \begin{cases}
          1-\sum_{\emph{k}\le q\text{, \emph{k} odd}} \binom{q}{\emph{k}} \cdot (1-a)^\emph{k}a^{q-\emph{k}}, & \text{if $\Phi(y^i)= A_q$} \\
          \sum_{\emph{k}\le q\text{, \emph{k} odd}} \binom{q}{\emph{k}} \cdot (1-a)^\emph{k}a^{q-\emph{k}}, & \text{if $\Phi(y^i)= B_q$}
    \end{cases}\nn\\
    &\triangleq \begin{cases}
          \alpha_q, & \text{if $\Phi(y^i)= A_q$}, \\
          1-\alpha_q, &\text {if $\Phi(y^i)= B_q$},
    \end{cases}
\end{align}
where $\alpha_q = \frac{1+(2a-1)^q}{2}$.

We now calculate $P(y_i|y^{i-1},s_0)$ for any possible output sequence $y^i$:
\begin{align}
    P(y_i=0|y^{i-1},s_0) &= \sum_{x_{i-1}^i} P(x_{i-1}|y^{i-1},s_0)P(x_i|x_{i-1})\nn\\&\quad\times P(y_i=0|x_i,x_{i-1})\nn\\
    &= \bar{\epsilon}\big{[}P_{X|X^-}(0|0)P(x_{i-1}=0|y^{i-1},s_0)\nn\\&\quad+P_{X|X^-}(1|1)P(x_{i-1}=1|y^{i-1},s_0)\big{]}\nn\\
    &= a\bar{\epsilon}, \label{eq:y_dist1}\\
    P(y_i=?|y^{i-1},s_0)
    &= \epsilon,\label{eq:y_dist2}\\
    P(y_i=1|y^{i-1},s_0) &=  \bar{a}\bar{\epsilon}P(x_{i-1}=0|y^{i-1},s_0)\nn\\
    &= \begin{cases} \bar{a}\bar{\epsilon}\alpha_q, & \text{if $\Phi(y^{i-1})=A_q$},\\
          \bar{a}\bar{\epsilon}(1-\alpha_q), & \text{if $\Phi(y^{i-1})=B_q$}.\label{eq:y_dist3}
    \end{cases}
\end{align}
To find the stationary distribution induced by the graph, we first calculate the transition probability $P_{Q|Q^-}$ as follows:
\begin{align}
    P_{Q|Q^-}(q_t|q_{t-1}) &= \sum_{x_{t-1}^t,y_t}P(x_{t-1},x_t,y_t,q_t|q_{t-1})\nn\\
        &= \sum_{x_{t-1}^t,y_t}P(x_{t-1}|q_{t-1})\nn\\&\times P(x_t|x_{t-1})P(y_t|x_t,x_{t-1})P(q_t|q_{t-1},y_t).\nn
\end{align}

Based on the graph symmetry and by using simple algebra, it follows that the stationary distribution that is induced by the transition probability $P_{Q|Q^-}$ is 
\begin{align}
    P_Q(A_q) = P_Q(B_q) =
          k\left(\frac{\epsilon}{1-a\bar{\epsilon}}\right)^q,
\end{align}
where $q\in\mathbb{N}\cup\{0\}$ and $k$ is a constant in $[0,1]$. Recall that the entries of the stationary distribution must sum to $1$:
\begin{align}
    \sum_{q\in\mathcal{Q}} P_Q(q) &= \sum_{i=0}^\infty 2k\left(\frac{\epsilon}{1-a\bar{\epsilon}}\right)^i \nn\\
    &\stackrel{(a)}= \frac{2k(1-a\bar{\epsilon})}{\bar{a}\bar{\epsilon}}, \nn
\end{align}
where $(a)$ follows by using the formula of a geometric series with a common ratio $\frac{\epsilon}{1-a\bar{\epsilon}}$. Hence,
\begin{align} \label{eq: dec_kcond}
    k = \frac{\bar{a}\bar{\epsilon}}{2(1-a\bar{\epsilon})}.
\end{align}

We can now find explicitly the lower bound in \eqref{eq: DEC_Markov_obj}:
\begin{align}
    \mathsf{C_{DEC}}&\geq \lim_{n\to\infty}\max_{P(x^n)\in\cal{P_{\mathrm{markov}}}}\frac{1}{n}I(X^n;Y^n|S_0=s_0) \nn\\
    &= \max_{a\in[0,1]}\sum_{q\in\mathcal{Q}}P_Q(q)H(Y_i|Q=q,S_0=s_0) - H_2(\epsilon)\nn\\
    &\stackrel{(a)}= \max_{a\in[0,1]}\sum_{q=0}^\infty \left(P_Q(A_q)+P_Q(B_q)\right)\nn\\&\qquad\qquad\times H_4\left(\epsilon,a\bar{\epsilon},\bar{a}\bar{\epsilon}\alpha_q,\bar{a}\bar{\epsilon}\bar{\alpha_q}\right) - H_2(\epsilon)\nn\\
    &= \max_{a\in[0,1]}\sum_{q=0}^\infty \bar{a}\bar{\epsilon}\left(P_Q(A_q)+P_Q(B_q)\right)H_2\left(\alpha_q\right) +\bar{\epsilon} H_2(a)\nn\\
    &= \max_{a\in[0,1]} \bar{\epsilon}H_2(a) + \frac{\bar{a}^2\bar{\epsilon}^2}{\epsilon}\sum_{q=0}^{\infty} \left(\frac{\epsilon}{1-a\bar{\epsilon}}\right)^{q+1}\nn\\&\qquad\qquad\times H_2\left(\frac{1-(2a-1)^q}{2}\right),
\end{align}
where $H_4(a_1,a_2,a_3,a_4) = -\sum_{i=1}^4 a_i\log_2(a_i)$ and (a) follows from \eqref{eq:y_dist1}--\eqref{eq:y_dist3}.

\end{proof}
\section{POST Channel --- Proof of Theorem \ref{Theorem: Post}} \label{app. POST}
\begin{proof}
The proof is based on the Markov $Q$-graph depicted in Fig. \ref{fig:1Markov} and the following optimized graph-based test distribution:
\begin{align*}
	R_{Y|Q}(0|0) = R_{Y|Q}(1|1) = (1+\bar{p}p^{\frac{p}{\bar{p}}})^{-1}.
\end{align*}

Define $\underline{z}$ as the pmf on $\mathcal{Q}$ that corresponds to the DP state. To simplify the notation, we denote $K\triangleq (1+\bar{p}p^{\frac{p}{\bar{p}}})^{-1}$. Further, since the vector $\underline{z}$ consists of only two components that sum to one, we then consider the DP state to be only the first component and denote it by $z$.

According to the DP formulation, the next DP state is computed as
\begin{align*}
F(z,u) &=
\begin{cases} 
      z+\bar{p}\cdot\left(1-z\right), & u = 0, \\
      p\cdot z, & u = 1.
\end{cases} 
\end{align*}
We now calculate the reward function explicitly. When $u=0$ the reward is
\begin{align*}
&g(z,u=0) \nn\\&= z\log_2\left(\frac{1}{K}\right)+\bar{z}\cdot\left[\bar{p}\log_2\left(\frac{\bar{p}}{1-K}\right)+p\log_2\left(\frac{p}{K}\right)\right],
\end{align*}
and when $u=1$ the reward is
\begin{align*}
&g(z,u=1) \nn\\&=  z\cdot\left[p\log_2\left(\frac{p}{K}\right)+\bar{p}\log_2\left(\frac{\bar{p}}{1-K}\right)\right] + \bar{z}\log_2\left(\frac{1}{K}\right).
\end{align*}

Recall that to solve the Bellman equation, one should identify a scalar $\rho\in\mathbb{R}$ and a function $h:\mathcal{Z}\to\mathbb{R}$ such that
\begin{align} \label{POST_bellman}
&\rho+h(z) \nn\\&= \max_u g(z,u)+h(F(z,u)), \nn\\&= \max_u
\begin{cases} 
      g(z,u=0) + h(z+\bar{p}\bar{z}), & u = 0, \\
      g(z,u=1) + h(pz), & u = 1,
 \end{cases} 
\end{align}
for all $z\in\mathcal{Z}$ and $p\in[0,1]$. In the following, we show that $\rho^* = \log_2\left(\frac{1}{K}\right)$ and the function
\begin{align*}
h^*(z) &= z\log_2\left(\frac{p\bar{K}}{\bar{p}K}\right)+\frac{(1-z)\log_2(p)}{1-p}, 
\end{align*}
solves \eqref{POST_bellman}.

Let us assume that the optimal policy is given by $u^* = 0$ for all $z\in\mathcal{Z}$. Accordingly, by using simple algebra, it follows that the right-hand side of \eqref{POST_bellman} is
\begin{align} \label{eq. POST_reward0}
&g(z,0)+h(F(z,0)) \nn\\&= \log_2\left(\frac{1}{K}\right) + z\log_2\left(\frac{p\bar{K}}{\bar{p}K}\right)+ \frac{\bar{z}}{\bar{p}}\log_2(p).
\end{align}
Further, for any $z\in\mathcal{Z}$, we note that \eqref{eq. POST_reward0} is exactly equal to the left-hand side of the Bellman equation. Therefore, assuming $u^*=0$, the Bellman equation is satisfied.
We will now verify that $u^*=0$. Again, by using simple algebra, we get
\begin{align*}
&g(z,1)+h(F(z,1)) \nn\\&= \log_2\left(\frac{1}{K}\right) + \frac{\bar{z}}{\bar{p}}\log_2(p) + 2zp\log_2\left(\frac{p\bar{K}}{\bar{p}K}\right)\nn\\&\quad+z\log_2\left(\frac{p\bar{p}K}{\bar{K}}\right),
\end{align*}
and therefore,
\begin{align*}
&\left[g(z,0)+h(F(z,0))\right] - \left[g(z,1)+h(F(z,1))\right] \nn\\&= z\log_2\left(\frac{p\bar{K}}{\bar{p}K}\right) - 2zp\log_2\left(\frac{p\bar{K}}{\bar{p}K}\right) - z\log_2\left(\frac{p\bar{p}K}{\bar{K}}\right)\\
&= 0.
\end{align*}
This implies that, for any choice of the action, the right-hand side of \eqref{POST_bellman} is the same. Therefore, the assumption that $u^*=0$ holds.
\end{proof}

\section*{Acknowledgment}
The authors would like to thank the Associate Editor and the anonymous reviewers for their valuable and constructive comments, which helped to improve this paper.
\end{appendices}
\bibliography{ref}
\bibliographystyle{IEEEtran}
\begin{IEEEbiographynophoto}{Bashar Huleihel}
(Student Member, IEEE) received the B.Sc. and M.Sc. degrees in electrical and computer engineering from the Ben-Gurion University of the Negev, Israel, in 2017 and 2020, respectively. He is currently pursuing the Ph.D. degree in electrical and computer engineering at the same institution. His research interests include information theory and machine learning.
\end{IEEEbiographynophoto}

\begin{IEEEbiographynophoto}{Oron Sabag}
(Member, IEEE) received the B.Sc. (cum laude), the M.Sc. (summa cum laude) and the Ph.D. in Electrical and Computer Engineering from the Ben-Gurion University of the Negev, Israel, in 2013, 2016 and 2019, respectively. He is currently a postdoctoral fellow with the Department of Electrical Engineering at Caltech. His research interests include control theory, information theory and reinforcement learning. 

He is a recipient of several awards, among them are ISEF postdoctoral fellowship, Lachish Fellowship, ISIT-2017 best student paper award, SPCOM-2016 best student paper award, the Feder Family Award for outstanding research in communications and the Kaufman award.
\end{IEEEbiographynophoto}

\begin{IEEEbiographynophoto}{Haim Permuter}
(Senior Member, IEEE) received the B.Sc. and M.Sc. degrees (summa cum laude) in electrical and computer engineering from Ben-Gurion University of the Negev, Israel, in 1997 and 2003, respectively, and the Ph.D. degree in electrical engineering from Stanford University, Stanford, CA, USA, in 2008. From 1997 to 2004, he was an Officer with the Research and Development Unit of the Israeli Defense Forces. Since 2009, he has been with the Department of Electrical and Computer Engineering, Ben-Gurion University of the Negev, where he is currently a Professor and the Luck-Hille Chair of electrical engineering. He also serves as the Head of the communication track in his department. He was a recipient of several awards, among them the Fullbright Fellowship, the Stanford Graduate Fellowship (SGF), the Allon Fellowship, and the U.S.–Israel Binational Science Foundation Bergmann Memorial Award. He has served on the editorial boards for the IEEE TRANSACTIONS ON INFORMATION THEORY from 2013 to 2016.
\end{IEEEbiographynophoto}
\newpage
\begin{IEEEbiographynophoto}{Navin Kashyap}
(Senior Member, IEEE) received the B.Tech. degree in Electrical Engineering from the Indian Institute of Technology, Bombay, in 1995, the M.S. degree in Electrical Engineering from the University of Missouri-Rolla in 1997, and the M.S. degree in Mathematics and the Ph.D. degree in Electrical Engineering from the University of Michigan, Ann Arbor, in 2001. From November 2001 to November 2003, he was a postdoctoral research associate at the University of California, San Diego. From 2004 to 2010, he was on the faculty of the Department of Mathematics and Statistics at Queen’s University, Kingston, Ontario. In January 2011, he joined the Department of Electrical Communication Engineering at the Indian Institute of Science, where is currently a Professor. His research interests lie primarily in the application of combinatorial and probabilistic methods in information and coding theory.

Prof. Kashyap was appointed as a Distinguished Lecturer of the IEEE Information Theory Society for 2017–2018. He served on the editorial board of the IEEE TRANSACTIONS ON INFORMATION THEORY during the period 2009–2014. He is at present an Associate Editor for the SIAM Journal on Discrete Mathematics and for the journal Cryptography and Communications (Springer).
\end{IEEEbiographynophoto}
\vspace{20cm}
\begin{IEEEbiographynophoto}{Shlomo Shamai}
(Life Fellow, IEEE) is currently with the Department of Electrical Engineering, Technion–Israel Institute of Technology, where he is also a Technion Distinguished Professor, and holds the William Fondiller Chair of Telecommunications. He is also an URSI Fellow, a member of the Israeli Academy of Sciences and Humanities, and a Foreign Member of the U.S. National Academy of Engineering. He was a recipient of the 2011 Claude E. Shannon Award, the 2014 Rothschild Prize in Mathematics/Computer Sciences and Engineering, and the 2017 IEEE Richard W. Hamming Medal. He was a co-recipient of the 2018 Third Bell Labs Prize for Shaping the Future of Information and Communications Technology. He was also a recipient of numerous technical and paper awards and recognitions of the IEEE (Donald G. Fink Prize Paper Award), Information Theory, Communications and Signal Processing Societies, and EURASIP. He is listed as a Highly Cited Researcher (Computer Science) for the years 2004, 2005, 2006, 2007, 2008, and 2013. He has served as an Associate Editor for the Shannon Theory of the IEEE TRANSACTIONS ON INFORMATION THEORY. He has also served twice on the Board of Governors for the Information Theory Society. He has also served on the Executive Editorial Board for the IEEE TRANSACTIONS ON INFORMATION THEORY, the IEEE Information Theory Society Nominations and Appointments Committee, and the IEEE Information Theory Society, Shannon Award Committee.
\end{IEEEbiographynophoto}
\end{document}